\documentclass[12pt]{article}
\usepackage[letterpaper,top=1.5cm, bottom=1.8cm, left=1.8cm, right=1.8cm]{geometry}

\usepackage{amssymb}
\usepackage{amsmath}
\usepackage{amsfonts}
\usepackage{amsthm}
\usepackage{latexsym}
\usepackage[usenames]{color} 
\usepackage{bm}
\usepackage{bbm}
\usepackage[normalem]{ulem}
\usepackage{overpic}
\usepackage{graphicx} 
\usepackage{mathrsfs}

\usepackage{exscale}

\def\Om{\Omega}
 \def\ua{\uparrow}

 \def\wh{\widehat}
 \def\wt{\widetilde}

\def\bbar{\overline}
\def\bR{\mathbb{R}}

\def\eps{\varepsilon}
\def\Om{\Omega}

\def\ignore#1{}

\def\cC{{\cal C}}

\def\cF{{\cal F}}

\def\cX{{\cal X}}

\def\<{\langle}
\def\>{\rangle}

\def\ua{\uparrow}

\def\cC{\mathscr C}

\def\bR{\mathbb R}
\def\cX{\mathscr X}
\def\cF{\mathscr F}
\def\bP{\mathbb P}
\def\bE{\mathbb E}
\def\bN{\mathbb N}
\def\bT{\mathbb T}

\def\eps{\varepsilon}

\def\Id{\,\text{\rm Id}}

\newtheorem{theorem}{Theorem}[section]
\newtheorem{proposition}[theorem]{Proposition}
\newtheorem{lemma}[theorem]{Lemma}

\theoremstyle{definition}
\newtheorem{definition}[theorem]{Definition}

\newtheorem{remark}[theorem]{Remark}

\def\wt{\widetilde}
\def\bbar{\overline}

 \newcommand{\1}{{\rm 1\hspace*{-0.4ex}
\rule{0.1ex}{1.52ex}\hspace*{0.2ex}}}
\def\Ind#1{\1_{_{\scriptstyle #1}}}

\begin{document}
\title{{\LARGE A  market impact  game under transient price impact }}
\author{\normalsize Alexander Schied\footnote{Department of Statistics and Actuarial Science, University of Waterloo, and Department of Mathematics, University of Mannheim, Email: {\tt alex.schied@gmail.com}} \qquad Tao Zhang\setcounter{footnote}{3}\footnote{Department of Mathematics, University of Mannheim, Email: {\tt taozhang.de@gmail.com}\hfill\break
The authors gratefully acknowledge financial support by Deutsche Forschungsgemeinschaft (DFG) through Research Grants SCHI/3-1 and SCHI/3-2.}}

\date{}

\maketitle

\begin{abstract}We consider a Nash equilibrium between two high-frequency traders in a simple market impact model with transient price impact and additional quadratic transaction costs. Extending a result by Sch\"oneborn (2008), we prove existence and uniqueness of the Nash equilibrium and  show that for small transaction costs the high-frequency traders engage in a \lq\lq hot-potato game", in which the same asset position is sold back and forth. We then identify a critical value for the size  of the transaction costs above which all oscillations disappear and strategies become buy-only or sell-only. Numerical simulations show that for both traders  the expected costs can be lower with transaction costs than without. Moreover, the costs can increase with the trading frequency if there are no transaction costs, but decrease with the trading frequency if transaction costs are sufficiently high. We argue that these effects occur due to the need of protection against predatory trading in the regime of low transaction costs. 
\end{abstract}
 
 \medskip



\section{Introduction}\label{introsection}

According to the Report~\cite{SEC} by CFTC and SEC  on the Flash Crash of May 6, 2010, the events that lead to the Flash Crash included a large sell order of E-Mini S\&P 500 contracts:
\begin{quote}
 \dots a large Fundamental Seller (\dots) initiated a program to sell a total of 75,000 E-Mini contracts (valued at approximately \$4.1 billion). \dots [On another]  occasion it took more than 5 hours for this large trader to execute the first 75,000 contracts of a large sell program. However, on May 6, when markets were already under stress, the Sell Algorithm chosen by the large Fundamental Seller to only target trading volume, and not price nor time, executed the sell program extremely rapidly in just 20 minutes.
\end{quote}
The  report~\cite{SEC} furthermore suggests that a \lq\lq hot-potato game" between high-frequency traders (HFTs) created artificial trading volume that at least contributed to the acceleration of the Fundamental Seller's trading algorithm:
\begin{quote}
\dots\  HFTs began to quickly buy and then resell contracts to each other---generating a \lq\lq hot-potato" volume effect as the same positions were rapidly passed back and forth. Between 2:45:13 and 2:45:27, HFTs traded over 27,000 contracts, which accounted for about 49 percent of the total trading volume, while buying only about 200 additional contracts net.
\end{quote} 
See also Kirilenko, Kyle, Samadi, and Tuzun~\cite{Kirilenko} and Easley, L\'opez da Prado, and O'Hara~\cite{Easley2011} for additional background.  

Sch\"oneborn~\cite{Schoeneborn} observed that the equilibrium strategies of two competing economic agents, who trade sufficiently fast in a simple market impact model with exponential decay of price impact, can exhibit strong oscillations. These oscillations have a striking similarity with the \lq\lq
hot-potato game" mentioned in~\cite{SEC} and~\cite{Kirilenko}.   In each trading period, one agent sells a large asset position to the other agent and buys a similar position back in the next period. The intuitive reason for this hot-potato game is to protect against possible \emph{predatory trading} by the other agent. Here, predatory trading refers to the exploitation of the drift generated by  the price impact of another agent. For instance, if the other agent is selling assets over a certain time interval, predatory trading would consist in  shortening the asset at the beginning of  the time interval and buying back when prices have depreciated through the sale of the other   agent. Such strategies are \lq\lq predatory" in the sense that their price impact decreases the revenues of the other agent and thus generate profit at the other agent's expense.

In this paper, we continue the investigation of the \lq\lq
hot-potato game". Our first contribution is to extend the result of Sch\"oneborn~\cite{Schoeneborn} by identifying a unique Nash equilibrium for two competing agents within a larger class of adaptive trading strategies, for general decay kernels,  and by giving an explicit formula for the equilibrium strategies. This explicit formula  will be the starting point for our further mathematical and numerical analysis of the Nash equilibrium. 
Another new feature of our approach is the addition of  quadratic transaction costs, which can be thought of temporary price impact in the sense of~\cite{BertsimasLo,AlmgrenChriss2} or as a transaction tax. The main goal of our paper is to study the impact of these additional transaction costs on equilibrium strategies. 
Theorem~\ref{nonosc thm}, our main result, precisely identifies a critical threshold $\theta^*$ for the size $\theta$ of these transaction costs at which all oscillations disappear. That is, for transactions $\theta\ge \theta^*$ certain  \lq\lq fundamental" equilibrium strategies consist exclusive of all buy trades or of all sell trades. For $\theta<\theta^*$, the \lq\lq fundamental" equilibrium strategies will contain both buy and sell trades when the decay of price impact in between two trades is sufficiently small.

In addition, numerical simulations will exhibit some rather striking properties of equilibrium strategies. They reveal, for instance, that the expected costs of both agents can be a \emph{decreasing} function of $\theta\in[0,\theta_0]$  when trading speed is sufficiently high. As a result, both agents can carry out their respective trades at a \emph{lower cost} when there are transaction costs, compared to the situation without transaction costs. Even more interesting is the behavior of the costs as a function of the trading frequency. We will see that,  for $\theta=\theta^*$, a higher trading speed can  \emph{decrease} expected trading costs, whereas the costs typically \emph{increase} for sufficiently small $\theta$. In particular the latter effect is surprising, because at first glance a higher trading frequency suggests that one has greater flexibility  in the choice of a strategy  and hence can become more cost efficient. So why are the costs then increasing in the trading frequency? 
We will argue that the intuitive reason for this effect is that a higher trading frequency results in greater possibilities for predatory trading by the competitor and thus requires taking additional measures of protection against predatory trading.  Some of these numerical observations have meanwhile been derived mathematically in our follow-up paper~\cite{SchiedStrehleZhang}, which has E.~Strehle as additional coauthor.

This paper builds on several research developments in the existing literature. First, there are several papers on predatory trading  such as Brunnermeier and Pedersen~\cite{BrunnermeierPedersen}, Carlin et al.~\cite{Carlinetal},  Sch\"oneborn and Schied~\cite{SchoenebornSchied}, and  the authors~\cite{SchiedZhangCARA} dealing with Nash equilibria for several agents that are active in a market model with temporary and permanent price impact.  A discrete-time market impact game with asymmetric information was analyzed by Moallemi et al.~\cite{Moallemietal}. In contrast to these previous studies, the transient price impact model we use here goes back to Bouchaud et al.~\cite{Bouchaudetal} and Obizhaeva and Wang~\cite{ow}. It was further developed in~\cite{AFS1,AFS2, Gatheral, ASS, PredoiuShaikhetShreve}, to mention only a few related papers. As first observed in~\cite{Schoeneborn}, the qualitative features of Nash equilibria for transient price impact differ dramatically from those obtained in~\cite{Carlinetal,SchoenebornSchied,SchiedZhangCARA}  for an Almgren--Chriss setting.   We refer to~\cite{GatheralSchiedSurvey,Lehalle} for  recent surveys on the price impact literature and extended bibliographies. Among the utilized mathematical tools are the theory of $M$-matrices~\cite{BermanPlemmons}, the correspondence between the inverses of triangular Toeplitz matrices and reciprocals of power series~\cite{Trench3}, and Kaluza's sign criterion for reciprocal power series~\cite{Kaluza,Szego}.

The paper is organized as follows. In Section~\ref{modeling section} we explain our modeling framework. The existence and uniqueness theorem for  Nash equilibria is stated in Section~\ref{Nash equilibrium section}. In Section~\ref{hot potato section} we analyze the oscillatory behavior of equilibrium strategies. Here we will also state our main result, Theorem~\ref{nonosc thm}, on the critical threshold for the disappearance of oscillations.  Our numerical results and their interpretation are presented in Section~\ref{hot potato section} and Section~\ref{transaction tax section}. Particularly,   Section~\ref{transaction tax section} contains the simulations for the behavior of the costs as a function of transaction costs and of trading frequency in the cases with and without transaction costs. 
The proofs of our results are given in Section~\ref{Proof Section}. We conclude in Section~\ref{Conclusion section}.

\section{Statement of results}\label{Results Section}

\subsection{Modeling framework}\label{modeling section}
We consider two financial agents, $X$ and $Y$, who are active in a  market impact model for one risky asset. Market impact will be transient and modeled  as in~\cite{ASS}; see also~\cite{Bouchaudetal,ow,AFS2,Gatheral,PredoiuShaikhetShreve} for closely related or earlier versions of this model, which is sometimes called a propagator model.
When none of the two agents is active, asset prices are described by a right-continuous martingale\footnote{The martingale assumption is natural from an economic point of view, because we are interested here in high-frequency trading over short time intervals $[0,T]$. See also the  discussions in~\cite{ASS,KSS} for additional arguments.} $S^0=(S^0_t)_{t\ge0}$ on a filtered probability space $(\Om,(\cF_t)_{t\ge0},\cF,\bP)$, for which $\cF_0$ is $\bP$-trivial. The process $S^0$ is often called the unaffected price process. 
Trading takes place at the discrete trading times of a \emph{time grid} $\bT =\{t_0,t_1,\dots, t_N\}$, where $0=t_0<t_1<\cdots<t_N= T$. Both agents are assumed to use trading strategies that are admissible in the following sense.

\begin{definition}\label{Strategies def}Suppose that a time grid $\bT =\{t_0,t_1,\dots, t_N\}$ is given. An \emph{admissible trading strategy} for $\bT$ and $Z_0\in\bR$ is a vector $\bm  \zeta=(\zeta_0,\dots,\zeta_N)$ of random variables such that 
\begin{enumerate}
\item each $\zeta_i$ is $\cF_{t_i}$-measurable and bounded, and 
\item $\displaystyle Z_0=\zeta_0+\cdots+\zeta_N$ $\bP$-a.s.
\end{enumerate}
The set of all admissible strategies for given $\bT$ and $Z_0$ is denoted by $\cX(Z_0,\bT)$.
\end{definition} 

For $\bm \zeta\in\cX(Z_0,\bT)$, the value of $\zeta_i$ is taken as the number of shares traded  at time $t_i$, with a positive sign indicating a sell order and a negative sign indicating a purchase. 
Thus, the requirement (b) in the preceding definition 
can be interpreted by saying that $Z_0$ is the inventory of the agent at time $0=t_0$ and that by time $t_N=T$ (e.g., the end of the trading day) the agent must have a zero inventory. The assumption that each $\zeta_i$ is bounded can  be made without loss of generality from an economic point of view. 

When the two agents $X$ and $Y$ apply respective strategies $\bm \xi\in\cX(X_0,\bT)$ and $\bm \eta\in\cX(Y_0,\bT)$,  the asset price is given by 
\begin{equation}\label{price process discrete}
S^{\bm \xi,\bm \eta}_t=S^0_t-\sum_{t_k<t}G(t-t_k)(\xi_k+\eta_k),
\end{equation}
where  $G:\bR_+\to\bR_+$ is a function called the \emph{decay kernel}.
Thus, at each time $t_k\in\bT$, the combined trading activities of the two agents move the current price by the amount $-G(0)(\xi_k+\eta_k)$. At a later time $t>t_k$, this price impact will have changed to $-G(t-t_k)(\xi_k+\eta_k)$. From an economic point of view it would be reasonable to assume that $G$ is nonincreasing, but this assumption is not essential for our results to hold mathematically. But we do assume throughout this paper that  the function $t\mapsto G(|t|)$ is strictly positive definite in the sense of Bochner: For all $n\in\bN$, $t_1,\dots,t_n\in\bR$, and $x_1,\dots,x_n\in\bR$ we have
\begin{align}\label{G strictly positive definite}
\sum_{i,j=1}^nx_ix_jG(|t_i-t_j|)\ge0,\quad\text{with equality  if and only if $x_1=\dots=x_n=0$.}
\end{align}
As observed in~\cite{ASS}, this assumption rules out the existence of price manipulation strategies in the sense of Huberman and Stanzl~\cite{HubermanStanzl}. It is satisfied as soon as $G$ is convex, nonincreasing, and nonconstant; see, e.g.,~\cite[Proposition 2]{ASS} for a proof.

Let us now discuss the definition of the liquidation costs incurred by each agent.  When only one agent, say $X$, places a nonzero order at time $t_k$, then we are in the situation of~\cite{ASS} and the price is moved linearly from $S^{\bm \xi,\bm \eta}_{t_k}$ to $S^{\bm \xi,\bm \eta}_{t_k+}:=S^{\bm \xi,\bm \eta}_{t_k}-{G(0)}\xi_k$.  The order  $\xi_k$ is therefore executed at the average price $\frac12(S^{\bm \xi,\bm \eta}_{t_k+}+S^{\bm \xi,\bm \eta}_{t_k})$ and consequently incurs the following expenses:
$$-\frac{1}2\big(S^{\bm \xi,\bm \eta}_{t_k+}+S^{\bm \xi,\bm \eta}_{t_k}\big)\xi_k=\frac{G(0)}2\xi_k^2-S_{t_k}^{\bm \xi,\bm \eta}\xi_k.
$$
Suppose now that the order $\eta_k$ of agent $Y$ is executed immediately after the order $\xi_k$. Then the price is moved linearly from $S^{\bm \xi,\bm \eta}_{t_k+}$ to $S^{\bm \xi,\bm \eta}_{t_k+}-{G(0)}\eta_k$, and the order of agent $Y$ incurs the expenses
$$-\frac12\big(S^{\bm \xi,\bm \eta}_{t_k+}+S^{\bm \xi,\bm \eta}_{t_k+}-{G(0)}\eta_k\big)\eta_k=\frac{G(0)}2\eta_k^2-S_{t_k}^{\bm \xi,\bm \eta}\eta_k+{G(0)}\xi_k\eta_k.
$$
So greater latency results in the additional cost term ${G(0)}\xi_k\eta_k$ for agent $Y$. Clearly, this term  appears in the expenses of agent $X$, if the roles  of $X$ and $Y$ are reversed. In the sequel, we are going to assume that none of the two agents has an advantage in latency over the other. Therefore, if both agents place nonzero orders at time $t_k$, execution priority is given to that agent who wins an independent coin toss.

In addition to the liquidation costs motivated above, we will also impose that each trade $\zeta_k$ incurs quadratic transaction costs of the form
$\theta \zeta_k^2$, where $\theta$ is a nonnegative parameter.   The assumption of quadratic transaction costs will be discussed at the end of Section~\ref{Nash equilibrium section}, after the statement of Theorem~\ref{th1}.

\begin{definition}\label{discrete strategies cost def}Suppose that $\bT =\{t_0,t_1,\dots, t_N\}$,  $X_0$ and $Y_0$ are given.  
Let furthermore $(\eps_i)_{i=0,1,\dots}$ be an i.i.d. sequence of Bernoulli\,$(\frac12)$-distributed random variables that are independent of $\sigma(\bigcup_{t\ge0}\cF_t)$. Then the \emph{costs of $\bm \xi\in\cX(X_0,\bT)$ given $\bm \eta\in\cX(Y_0,\bT)$} are defined as
\begin{equation}\label{costs def eq}
\cC_\bT(\bm \xi|\bm \eta)=X_0S_0^0+\sum_{k=0}^N\Big(\frac{G(0)}2\xi_k^2-S_{t_k}^{\bm \xi,\bm \eta}\xi_k+\eps_k{G(0)}\xi_k\eta_k+\theta\xi_k^2\Big)
\end{equation}
and the \emph{costs of $\bm \eta$ given $\bm \xi$} are
$$\cC_\bT(\bm \eta|\bm \xi)=Y_0S_0^0+\sum_{k=0}^N\Big(\frac{G(0)}2\eta_k^2-S_{t_k}^{\bm \xi,\bm \eta}\eta_k+(1-\eps_k){G(0)}\xi_k\eta_k+\theta\eta_k^2\Big).
$$
\end{definition}

The term $X_0S_0^0$ corresponds to the book value of the position $X_0$ at time $t=0$. If the position $X_0$ could be liquidated at book value, one would  incur the expenses $-X_0S_0^0$.  Therefore, the liquidation costs as defined in \eqref{costs def eq}  are the difference of the actual accumulated expenses, as represented  by the sum on the right-hand side of \eqref{costs def eq}, and the expenses for liquidation at book value. The following  remark provides further comments on our modeling assumptions.

\begin{remark}\label{bid-ask-rem}   The market impact model we are using here has often been linked to the placement of market orders in a block-shaped limit order book, 
and a bid-ask spread is sometimes added to the model so as to make this interpretation more feasible~\cite{ow,AFS1}. For  a strategy consisting exclusively of market orders, the bid-ask spread will lead to an additional fee that should be reflected in the corresponding cost functional. 
  In reality, however, most strategies will involve  a variety of different order types and one should think of the costs \eqref{costs def eq} as the costs \emph{averaged} over   order types, as is often done in the market impact literature. For instance, while one may have to pay the spread when placing a market   order, one essentially earns it back when a limit  order is executed. Moreover, high-frequency traders often have access to a variety of more exotic order types, some of which can  pay rebates when executed. It is also possible to use crossing networks or dark pools in which orders are executed at mid price. So, for a setup of  high-frequency trading, taking the bid-ask spread as zero in \eqref{price process discrete} is probably  more realistic than modeling every single order as a market order and to impose the fees. The existence of hot-potato games in real-world markets, such as the one quoted from~\cite{SEC} in Section~\ref{introsection}, can be regarded as an empirical justification of the zero-spread assumption, because such a trading behavior could never be profitable if each trader had to pay the full spread upon each execution of an order.   See also the end of Section~\ref{Nash equilibrium section} for a
  discussion on how to replace our quadratic transaction costs by piecewise linear ones.\end{remark}

\subsection{Nash equilibrium}\label{Nash equilibrium section}

We now consider agents who need to liquidate their current  inventory within a given time frame  and who are aiming to minimize the expected  costs over admissible strategies. The need for liquidation can arise due to various reasons. For instance, Easley, L\'opez da Prado, and O'Hara~\cite{Easley2011} argue that the toxicity of the order flow preceding the Flash Crash of May 6, 2010, has led the inventories of several high-frequency market makers to grow beyond their risk limits, thus forcing  them to unload their  inventories.

When just a single agent is considered, the minimization of the expected execution costs  is  a well-studied problem; we refer to~\cite{ASS} for an analysis within our current modeling framework. Here we are going to investigate the optimal strategies of our two agents, $X$ and $Y$, under the assumption that both have full knowledge of the other's strategy and maximize the expected costs of their strategies accordingly. In this situation, it is natural to define optimality through the following notion of a Nash equilibrium.

\begin{definition}\label{Nash Def}For given time grid $\bT$ and initial values $X_0$, $Y_0\in\bR$, a \emph{Nash equilibrium} is  a pair $(\bm \xi^*,\bm \eta^*)$ of strategies in $\cX(X_0,\bT)\times\cX(Y_0,\bT)$ such that 
$$\bE[\,\cC_\bT(\bm \xi^*|\bm \eta^*)\,]=\min_{\bm \xi\in\cX(X_0,\bT)}\bE[\,\cC_\bT(\bm \xi|\bm \eta^*)\,]\qquad\text{and}\qquad \bE[\,\cC_\bT(\bm \eta^*|\bm \xi^*)\,]= \min_{\bm \eta\in\cX(Y_0,\bT)}\bE[\,\cC_\bT(\bm \eta|\bm \xi^*)\,].
$$
\end{definition}

To state our formula for this Nash equilibrium, 
we need to introduce the following notation. For a fixed time grid $\bT=\{t_0,\dots, t_N\}$, we define the $(N+1)\times(N+1)$-matrix ${\Gamma}$ by 
\begin{equation}\label{G matrix eq}
{\Gamma}_{i,j}=G(|t_{i-1}-t_{j-1}|),\qquad i,j=1,\dots, N+1,
\end{equation}
and for $\theta\ge0$ we introduce
\begin{align}\label{Gamma theta}
\Gamma_\theta:=\Gamma+2\theta\Id.
\end{align}
 We furthermore define the lower triangular matrix $\wt {\Gamma}$  by 
\begin{align}\label{wt Gamma def}
\wt {\Gamma}_{ij}=\begin{cases}{\Gamma}_{ij}&\text{if $i>j$,}\\
\frac12G(0)&\text{if $i=j$,}\\
0&\text{otherwise.}
\end{cases}
\end{align}
Note that $\Gamma=\wt\Gamma+\wt\Gamma^\top$, where $\top$ denotes the transpose of a matrix or vector.
We will write ${\bm 1}$ for the vector $(1,\dots, 1)^\top\in\bR^{N+1}$. A strategy $\bm\zeta=(\zeta_0,\dots,\zeta_N)\in\cX(Z_0,\bT)$ will be identified with the ${(N+1)}$-dimensional random vector $(\zeta_0,\dots,\zeta_N)^\top$. Conversely, any vector $\bm z=(z_1,\dots, z_{N+1})^\top \in\bR^{N+1}$ can be identified with the deterministic strategy $\bm\zeta$ with $\zeta_k=z_{k+1}$.
We also define the two vectors
\begin{equation}\label{v and w def}\begin{split}
\bm v&=\frac{1}{{\bm 1}^\top(  {{\Gamma_\theta}}+  \wt{{{\Gamma}}})^{-1}{\bm 1}}({{\Gamma_\theta}}+  \wt{{{\Gamma}}})^{-1}{\bm 1}\\
\bm w&=\frac{1}{{\bm 1}^\top(  {{\Gamma_\theta}}-  \wt{{{\Gamma}}})^{-1}{\bm 1}}(  {{\Gamma_\theta}}-  \wt{{{\Gamma}}})^{-1}{\bm 1}.
\end{split}
\end{equation}
It will be shown in Lemma~\ref{positive definite lemma} below that the  matrices ${{\Gamma_\theta}}+  \wt{{{\Gamma}}}$ and ${{\Gamma_\theta}}-  \wt{{{\Gamma}}}$ are indeed invertible and that the denominators in \eqref{v and w def} are strictly positive under our assumption   \eqref{G strictly positive definite} that $G(|\cdot|)$ is strictly positive definite. Recall that we assume  \eqref{G strictly positive definite} throughout this paper.

In the case $G(t)=\gamma+\lambda e^{-\rho t}$ for constants $\gamma\ge0$ and $\lambda,\rho>0$, the existence of a unique Nash equilibrium in the class of deterministic strategies was established in Theorem 9.1 of~\cite{Schoeneborn}. Our subsequent Theorem~\ref{th1}   extends this result in a number of ways: we  allow for general positive definite decay kernels,  include transaction costs,  give an explicit form of the deterministic Nash equilibrium, and show that this Nash equilibrium is also the unique Nash equilibrium in the class of adapted strategies. Our explicit formula for the equilibrium strategies will be the starting point for our further mathematical and numerical analysis of the Nash equilibrium. Also our proof is different from the one in~\cite{Schoeneborn}, which  works only for the specific decay kernel $G(t)=\lambda e^{-\rho t}+\gamma$.

\begin{theorem}\label{th1}For any  strictly positive definite decay kernel $G$, time grid $\bT$, parameter $\theta\ge0$, and initial values $X_0$, $Y_0\in\bR$, there exists a unique  Nash equilibrium $(\bm \xi^*,\bm \eta^*)\in\cX(X_0,\bT)\times\cX(Y_0,\bT)$. The optimal strategies $\bm \xi^*$ and $\bm \eta^*$ are deterministic and given by
\begin{equation}\label{equilibrium strategies}
\begin{split}
\bm \xi^*&=\frac12(X_0+Y_0)\bm v+\frac12(X_0-Y_0)\bm w,\\
\bm \eta^*&=\frac12(X_0+Y_0)\bm v-\frac12(X_0-Y_0)\bm w.
\end{split}
\end{equation}
\end{theorem}

The formula \eqref{equilibrium strategies} shows that the vectors $\bm v$ and $\bm w$ form a basis for all possible equilibrium strategies. It follows that in analyzing the Nash equilibrium it will be sufficient to study the two cases  $\bm \xi^*=\bm v=\bm\eta^*$ for  $X_0=1=Y_0$ and $\bm \xi^*=\bm w=-\bm\eta^*$ for  $X_0=1=-Y_0$.

Let us now comment on our choice of quadratic transaction costs. Such quadratic transaction costs are often used to model   \lq\lq slippage" arising from temporary price impact; see~\cite{BertsimasLo,AlmgrenChriss2} and~\cite[Section 2.2]{Gatheral}. Nevertheless, proportional transaction costs might be more realistic in many situations, and so the question arises if our results will change when the quadratic transaction costs $\theta\xi_k^2$ are replaced by (piecewise) linear transaction costs. This question is  at least partially answered by the following result. It states that our quadratic transaction cost function can be replaced by proportional transaction costs in a neighborhood of the origin without affecting the Nash equilibrium. Since the main difference of quadratic and proportional transaction costs is their behavior at the origin, one may therefore guess  that similar results as obtained in the following sections for quadratic transaction costs might also hold for proportional transaction costs. 

\begin{proposition}\label{linear transaction costs prop}
In the context of Theorem~\ref{th1}, there exists a piecewise linear, increasing, convex, and continuous transaction cost function $\tau$ with $\tau(0)=0$ such that $(\bm \xi^*,\bm \eta^*)$ from~\eqref{equilibrium strategies} is a Nash equilibrium in $\cX(X_0,\bT)\times\cX(Y_0,\bT)$ for the the modified expected cost functional in which the quadratic transaction cost function $x\mapsto \theta x^2$ is replaced with $x\mapsto\tau(|x|)$.
\end{proposition}

The transaction cost function $\tau$ constructed in the preceding proposition is of the form 
$$\tau(|x|)=\theta_0|x|+\sum_{k=1}^{M}\theta_k(|x|-c_k)\Ind{[c_k,\infty)}(|x|)
$$
for certain coefficients $\theta_k>0$ and thresholds $0<c_1<\cdots c_{M}$. Transaction costs of this form can  model a transaction tax that is subject to tax progression. With such a tax, small orders, such as those  placed by small investors, are taxed at a lower rate than large orders, which may be placed with the intention of moving the market. 

\subsection{The hot-potato game}\label{hot potato section}

We now turn toward a qualitative analysis of the equilibrium strategies. By means of numerical simulations and the analysis of a particular example, Sch\"oneborn~\cite[Section 9.3]{Schoeneborn} 
observed that the equilibrium strategies may exhibit strong oscillations if  $\theta=0$, the time grid is equidistant, and $G$ is of the form $G(t)=\lambda e^{-\rho t}+\gamma $ for constants $\lambda,\rho>0$ and $\gamma\ge0$. As a matter of fact, numerical simulations, such as those presented in Figures~\ref{StrategiesFig} and~\ref{random grid figure}, suggest that such oscillations can be observed for a large class of decay kernels as soon as transaction costs vanish ($\theta=0$) and the time grid is sufficiently fine.  We refer to Remark~\ref{FinancialInterpretRemark} for a possible financial interpretation of the oscillations arising in the hot-potato game. For a single financial agent, however, optimal strategies will always be buy-only or sell-only for convex, nonincreasing decay kernels, which include those used in Figures~\ref{StrategiesFig} and~\ref{random grid figure} (see~\cite[Theorem 1]{ASS}). Therefore, the oscillations  in our two-agent setting that are observed in these figures  must necessarily result from the interaction of both agents.

It is intuitively  clear that increased transaction costs will penalize  oscillating strategies and thus lead to a smoothing of the equilibrium strategies. As a matter of fact, one can see in Figure~\ref{nonosc fig} that for $\theta=2$ all oscillations have disappeared so that equilibrium strategies are then buy-only or sell-only.  
 One can therefore wonder whether between $\theta=0$ and $\theta=2$ there might be a \emph{critical value}  $\theta^*$ at which all oscillations of $\bm v$ and $\bm w$ disappear, but below which oscillations are present. That is, for $\theta\ge\theta^*$ all equilibrium strategies should be either buy-only or sell-only, while for $\theta<\theta^*$   equilibrium strategies should contain  both buy and sell trades  (at least for certain values of $N$ and $T$).  
 The following theorem confirms that such a critical value $\theta^*$ does indeed exist.  
We can even determine its precise  value in case that  we are dealing with {equidistant time grids},
\begin{equation}\label{equidistant time grid eq}
\bT_N:=\Big\{\frac {kT}N\,\Big|\,k=0,1,\dots,N\Big\},\qquad N\in\bN.
\end{equation}
And we will be able to say even more in case $G$ is of the form 
\begin{equation}\label{exponential decay kernel}
G(t)=\lambda e^{-\rho t}+\gamma\qquad\text{for constants $\lambda,\rho>0$ and $\gamma\ge0$.}
\end{equation}
It is well known that this class of decay kernels satisfies our assumption \eqref{G strictly positive definite} (see, e.g.,~\cite[Example 1]{ASS}), and they are clearly log-convex.

\begin{theorem}\label{nonosc thm}Suppose that  $G$ is a continuous, positive definite, strictly positive, and log-convex decay kernel and that   $\bT_N$ denotes the equidistant time grid \eqref{equidistant time grid eq}. Then the following conditions are equivalent.
\begin{enumerate}
\item For every $N\in\bN $ and $T>0$, all components of $\bm w$ are nonnegative.
\item $\theta\ge\theta^*=G(0)/4$.
\end{enumerate}
If, moreover, $G$ is of the form \eqref{exponential decay kernel}, then conditions {\rm (a)} and {\rm (b)} are equivalent to:
\begin{enumerate}
\item[\rm (c)] For every $N\in\bN $ and $T>0$, all components of $\bm v$ are nonnegative.
\end{enumerate}
\end{theorem}

\medskip

In the case $\theta<\theta^*$, one can actually obtain some stronger results on the existence of oscillations in the vector $\bm w$. These are stated in the following two propositions. First, we deal with the oscillations of the signs of the last three trades of $\bm w$, which are present as soon as $\theta<\theta^*$ and the time grid is sufficiently fine. Recall that $\bm w$ completely determines the unique Nash equilibrium with initial conditions $X_0=-Y_0$. 

\medskip

\begin{proposition}\label{three oscillations prop}Suppose that $G$ is a continuous and positive definite decay kernel that is nonincreasing in a neighborhood of zero.  Then for $0\le \theta<\theta^*$ there exists $\delta>0$ such that for all time grids $\bT=\{t_0,t_1,\dots, t_N\}$ with $t_N-t_{N-1}<\delta$ and $t_{N-1}-t_{N-2}<\delta$, the last three components of  the vector $\bm w$ satisfy $w_{N+1}>0$, $w_N<0$, and $w_{N-1}>0$. \end{proposition}

The simulations in Figures~\ref{StrategiesFig} and~~\ref{random grid figure} show that for $\theta=0$ actually all components of the vectors $\bm w$ and $\bm v$ have oscillating signs. The following propositions establishes the existence of oscillations for $\bm w$ in the case of an exponential decay kernel and an equidistant time grid.

\begin{proposition}\label{oscillations prop}Suppose that $G$ is of the form $G(t)=\lambda e^{-\rho t}$ for constants $\lambda,\rho>0$ and that $\bT_N$ denotes the equidistant time grid \eqref{equidistant time grid eq} for some given $T>0$. Then there exists  $N_0\in\bN$ such that for each $N\ge N_0$  there exists $\delta>0$ so  that for $0\le\theta<\delta$  all entries of the vector $\bm w=(w_1,\dots, w_{N+1})$ are nonzero and have alternating signs.
\end{proposition}

We refer to the right-hand panel of Figure~\ref{StrategiesFig} for an illustration of the oscillations of the vector $\bm w$.  As shown in the left-hand panel of  the same figure, similar oscillations occur for the vector ${\bm v}$ and hence for equilibria with arbitrary initial conditions. The mathematical analysis for ${\bm v}$, however, is much harder than for ${\bm w}$, and at this time we are not able to prove a result that could be an analogue of Proposition~\ref{oscillations prop} for the vector $\bm v$. The existence of oscillations of ${\bm w}$ and ${\bm v}$ is also not limited to exponential decay kernels as can be seen from numerical experiments;  see Figure~\ref{random grid figure} for power law decay and a randomly generated, non-equidistant time grid. \begin{figure}
\centering
\begin{minipage}[b]{8cm}
\includegraphics[width=8cm]{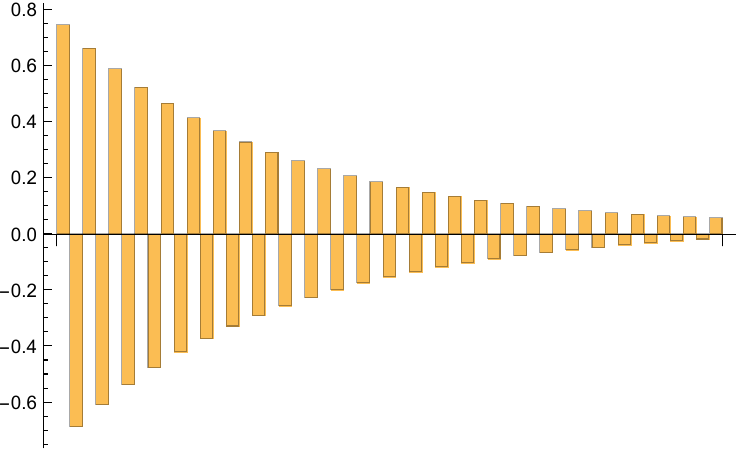}\\
\end{minipage}\qquad
\begin{minipage}[b]{8cm}
\includegraphics[width=8cm]{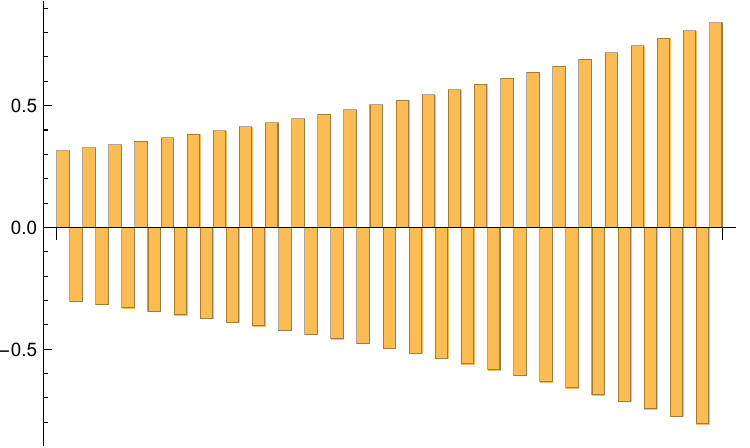}\\
\end{minipage}
 \caption{Vectors $\bm v$ (left) and $\bm w$ (right) for the equidistant time grid $\bT_{50}$, $G(t)=e^{-t}$,  $\theta=0$, and $T=1$. By \eqref{equilibrium strategies}, $(\bm v,\bm v)$ is the equilibrium  for $X_0=Y_0=1$, and $(\bm w,-\bm w)$ is the equilibrium  for $X_0=-Y_0=1$. Yet, some individual components of both $\bm v$ and $\bm w$ exceed  in either direction 60\% of the sizes of the initial positions $X_0$ and $Y_0$.}\label{StrategiesFig}
\end{figure}
\begin{figure}
\centering
\begin{minipage}[b]{8.0cm}
\includegraphics[width=8cm]{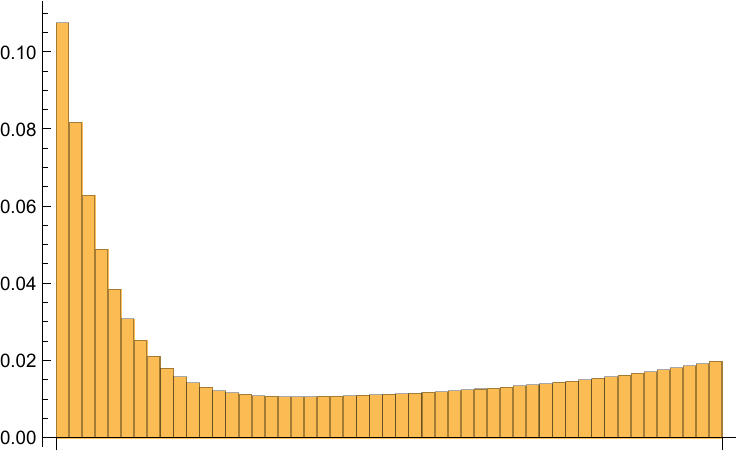}\\
\end{minipage}\qquad
\begin{minipage}[b]{8cm}
\includegraphics[width=8cm]{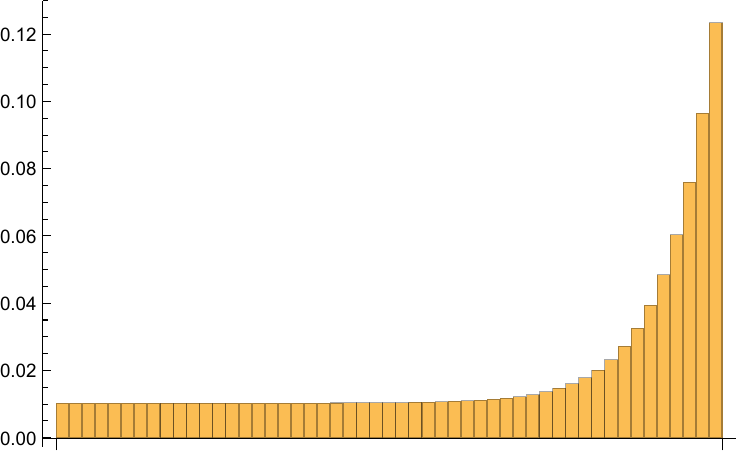}\\
\end{minipage}
 \caption{Vectors $\bm v$ (left) and $\bm w$ (right) for the equidistant time grid $\bT_{50}$, $G(t)=e^{-t}$,  $\theta=2$, and $T=1$. }\label{nonosc fig}
\end{figure}
\begin{figure}
\centering
\begin{minipage}[b]{8cm}
\includegraphics[width=8cm]{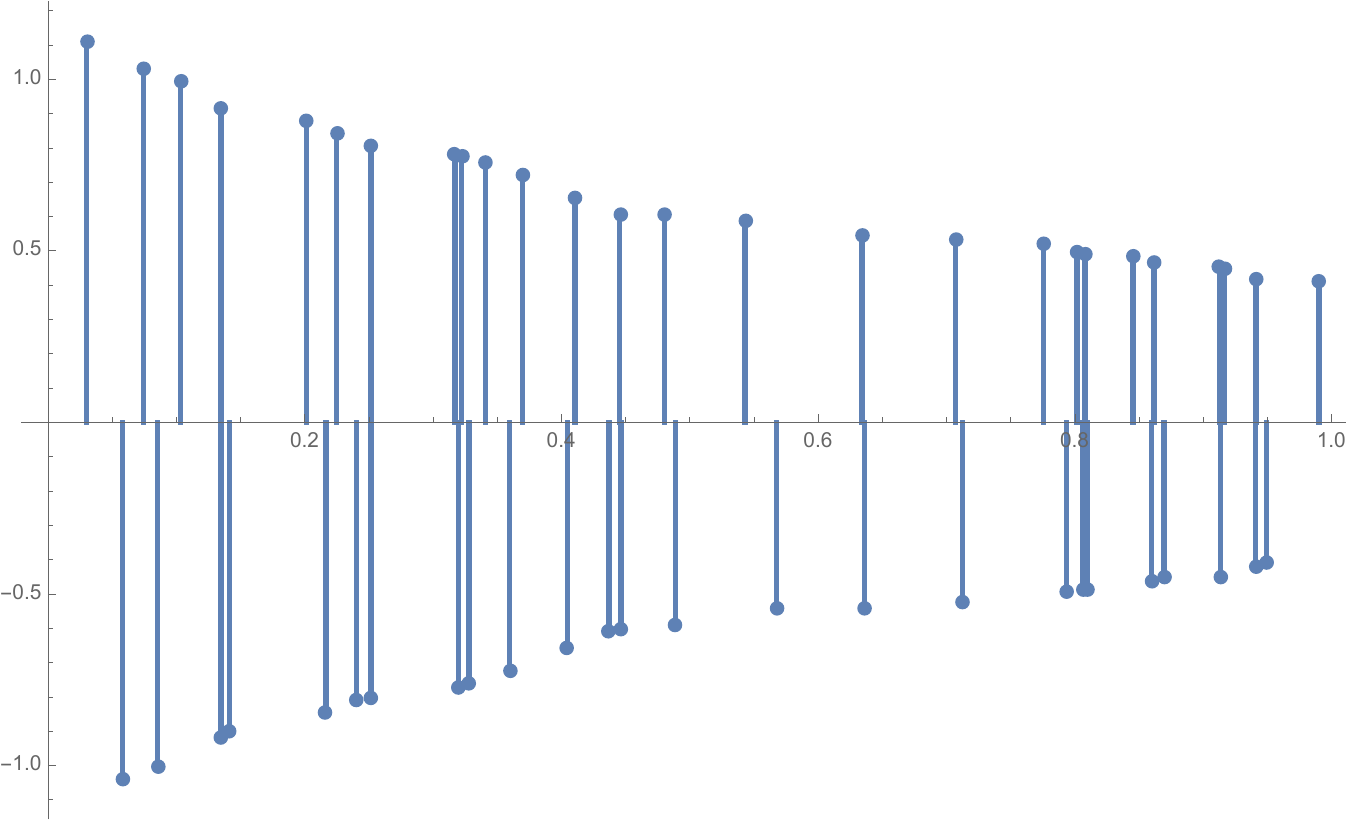}\\
\end{minipage}\qquad
\begin{minipage}[b]{8cm}
\includegraphics[width=8cm]{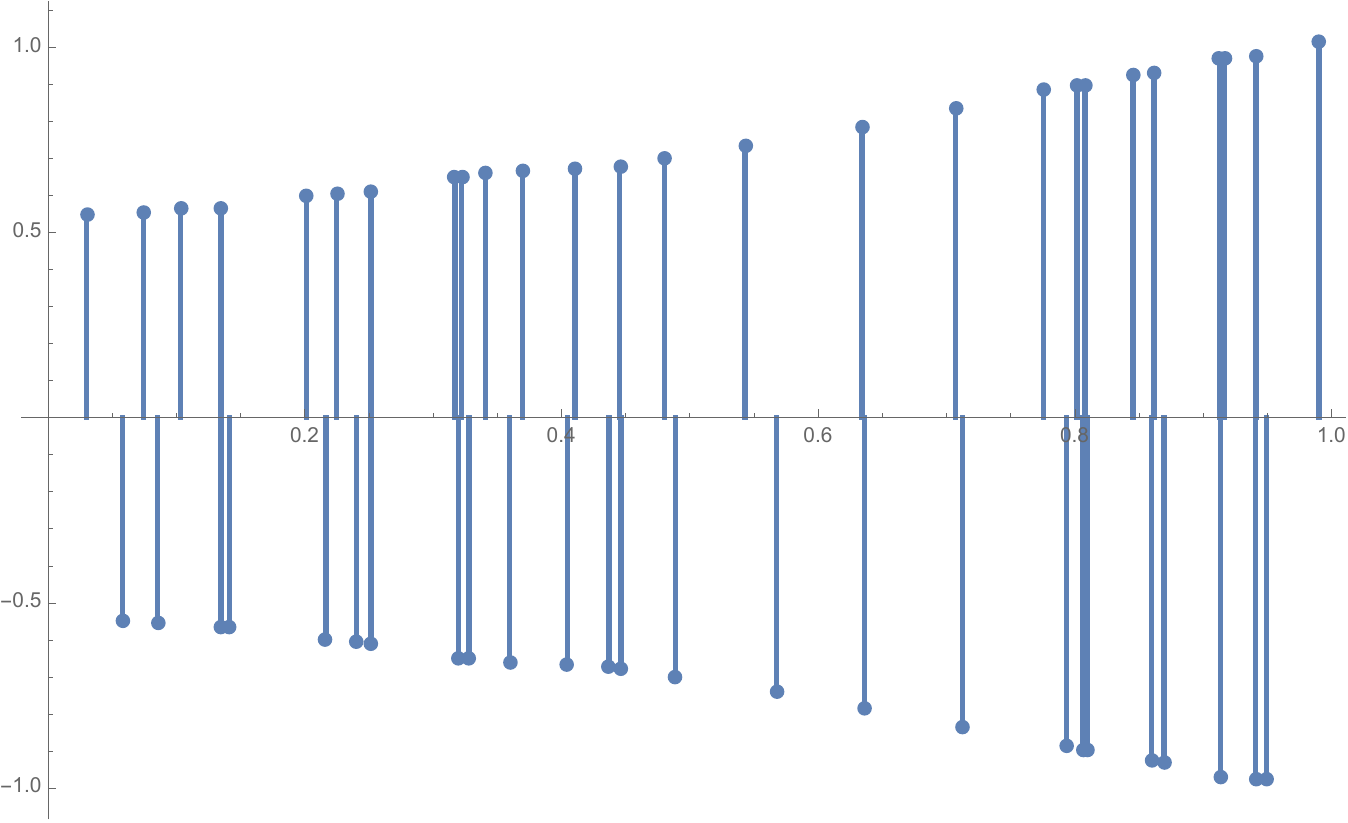}\\
\end{minipage}
 \caption{Vectors $\bm v$ (left) and $\bm w$ (right) for   power-law decay  $G(t)=1/\sqrt{1+t}$ and a time grid generated from 50 independent uniformly distributed random variables on.}
 \label{random grid figure}\end{figure}

\medskip

\begin{remark}\label{FinancialInterpretRemark} 
In this remark we will discuss a possible financial explanation for the oscillations of equilibrium strategies  observed  for small values of $\theta$. As mentioned above, the source for these oscillations must necessarily lie in the interaction between the two agents. As observed in previous studies on multi-agent equilibria in price impact models such as~\cite{BrunnermeierPedersen,Carlinetal,SchoenebornSchied}, the dominant form of interaction between two players is \emph{predatory trading}, which consists in the exploitation of price impact generated by another agent.  Such strategies are \lq\lq predatory" in the sense that they  generate profit  by simultaneously decreasing  the other agent's revenues. Since predators prey on the drift created by the price impact of a  large trade, protection against predatory trading requires the cancellation of previously created price impact. Under transient price impact, the price impact of an earlier trade, say $\zeta_0$, can be cancelled by  placing an order $\zeta_{ 1}$ of the opposite side. For instance, taking $\zeta_{1}:=-\zeta_0G(t_{1}-t_0)$ will completely  eliminate the  price impact of  $\zeta_0$ while the combined trades execute  a total of $\xi_0(1-  G(t_{1}-t_0))$ shares. In this sense, oscillating strategies can be understood as a protection against predatory trading by opponents (see also~\cite[p.150]{Schoeneborn}).
\end{remark}

\begin{remark}
Alfonsi et al.~\cite{ASS} discovered  oscillations for the trade execution strategies of a \emph{single} trader under transient price impact if price impact does not decay as a convex function of time. These oscillations, however, result from an attempt to exploit the delay in market response to a large trade, and they disappear if price impact decays as a convex function of time~\cite[Theorem 1]{ASS}. In particular, when there is just one agent active and $G$ is convex, nonincreasing, and nonconstant (which is, e.g., the case under assumption \eqref{exponential decay kernel}), then for each time grid $\bT$ there exists a unique optimal strategy, which is either buy-only or sell-only. When \eqref{exponential decay kernel} holds and $\theta=0$, this strategy is known explicitly; see~\cite{AFS1}.
\end{remark}

\begin{remark}Based on numerical simulations, we believe that the statements of Theorem~\ref{nonosc thm} and Proposition~\ref{oscillations prop} can probably be improved. Specifically, we conjecture that the equivalence between the conditions (a), (b), and (c)  in Theorem~\ref{nonosc thm} remains true for all positive definite decay kernels. Our current proofs, however, cannot be extended beyond our stated conditions. Specifically, the implication of (b)$\Rightarrow$(a) in Theorem~\ref{nonosc thm} exploits the Toeplitz structure of the upper triangular matrix $\Gamma_\theta-\wt\Gamma$, which only holds for equidistant time grids. We then use the fact that the inverse of a triangular Toeplitz matrix corresponds to the (formal)  reciprocal of a power series, and we use the celebrated Kaluza sign criterion~\cite{Kaluza,Szego} to determine the signs of this reciprocal power series. Here, the log-convexity of $G$ is essential. The proof of the implication (b)$\Rightarrow$(c) relies on the theory of $M$-matrices as presented in~\cite{BermanPlemmons}. In particular, we rely on the fact that the matrix $\Gamma^{-1}(\wt\Gamma+\frac12\text{Id})$ is a non-singular $M$-matrix for $G(t)=e^{-\rho t}$ (Lemma~\ref{lm M aux}), which is no longer  true, e.g., for power law decay $G(t)=1/(1+t)^p$ with $p>0$. Similarly, the proof of Proposition~\ref{oscillations prop} exploits the fact that the upper triangular  matrix $\Gamma-\wt\Gamma$ can be inverted explicitly if the time grid is equidistant and $G(t)=e^{-\rho t}$. Surprisingly, although the matrix $\Gamma$ has an explicit inverse for any time grid if $G(t)=e^{-\rho t}$ (see~\cite[Theorem 3.4]{AFS1}), the structure of $(\Gamma-\wt\Gamma)^{-1}$ becomes quite involved if the time grid is not equidistant. Already for equidistant time grids, the same can be said of the matrix $(\Gamma+\wt\Gamma)^{-1}$, which is needed to compute the vector $\bm v$. \end{remark}

\subsection{The impact of transaction costs and trading frequency on the expected costs}\label{transaction tax section}

Due to our explicit formulas \eqref{v and w def} and \eqref{equilibrium strategies}, it is easy to analyze the Nash equilibrium numerically.  These numerical simulations exhibit several striking effects in regards to monotonicity properties of the expected costs. 

In Figure~\ref{Monotonicity Cost fig} we have plotted the expected costs
$\bE[\,\cC_{\bT_N}(\bm\xi^*|\bm\eta^*)\,]= \bE[\,\cC_{\bT_N}(\bm\eta^*|\bm\xi^*)\,]$ for $X_0=Y_0$, $G(t)=e^{-t}$, and  $T=1$ as a function of the trading frequency, $N$. The first observation one probably makes when looking at this plot is the fact that for $\theta=0$ the expected costs exhibit a sawtooth-like pattern; they alternate between two increasing trajectories, depending on whether $N$ is odd or even. These alternations are due to the oscillations of the optimal strategies, which also alternate with $N$.  As can be seen from the figure, the sawtooth pattern  essentially disappears already for very small values of $\theta$ such as for  $\theta=0.01$.

A   more interesting observation is the fact that for  $\theta=0$, $\theta=0.01$, and $\theta=0.1$ the expected costs $\bE[\,\cC_{\bT_{2N}}(\bm\xi^*|\bm\eta^*)\,]$ (or alternatively $\bE[\,\cC_{\bT_{2N+1}}(\bm\xi^*|\bm\eta^*)\,]$) are \emph{increasing} in $N$. This fact is surprising because a higher trading frequency should  normally lead to a larger class of admissible strategies. As a result, traders have greater flexibility in choosing  a   strategy and  in turn should be able to pick  more cost efficient strategies. So why are the costs then increasing in $N$? The intuitive explanation is that a higher trading frequency increases also the possibility for the competitor to conduct predatory strategies at the expense of the other agent (see Remark~\ref{FinancialInterpretRemark}). In reaction,  this other agent needs  to take stronger protective measures against predatory trading. As discussed in  Remark~\ref{FinancialInterpretRemark}, protection against predatory trading can be obtained by erasing (part of) the previously created price impact  through  placing an order of the opposite side. The result is an oscillatory strategy, whose expected costs increase with the number of its oscillations. 

Still in Figure~\ref{Monotonicity Cost fig}, the expected costs $\bE[\,\cC_{\bT_N}(\bm\xi^*|\bm\eta^*)\,]$ for the case $\theta=\theta^*=0.25$ exhibit a very different behavior. They no longer alternate in $N$ and are \emph{decreasing} as a function of the trading frequency. The intuitive explanation is that transaction costs of size $\theta^*=0.25$ discourage predatory trading to a large extend,  so that agents can now benefit from a higher trading frequency and pick ever more cost-efficient strategies as $N$ increases.

The most surprising observation in Figure~\ref{Monotonicity Cost fig} is the fact that for sufficiently large $N$ the expected costs for  $\theta>0$  fall  below the expected costs for $\theta=0$. That is, for sufficiently large trading frequency, \emph{adding transaction costs can decrease the expected  costs of all market participants} (recall that for $X_0=Y_0$ both agents have the same optimal strategies and, hence, the same expected costs). 
This fact is further illustrated in Figure~\ref{costs(theta) fig}, which exhibits a very steep initial decrease of the  expected costs as a function of $\theta$. After a minimum of the expected costs is reached at $\theta\approx 0.06$, there is a slow and steady increase of the costs with an approximate slope of 0.002.

The key to understanding the behavior of expected equilibrium costs as a function of trading frequency and transaction costs rests in the interpretation of the oscillations in equilibrium strategies as a protection against predatory trading by the opponent (see   Remark~\ref{FinancialInterpretRemark}). Note that a predatory trading strategy is necessarily  a \lq\lq round trip", i.e., a strategy with zero inventory at $t=0$ and $T=0$ (the strategy of a  predatory trader with nonzero initial position would consist of a superposition of a predatory round trip and a liquidation strategy for the initial position). It therefore must consist of a buy and a sell component and is hence   stronger penalized by an increase in transaction costs than a buy-only or sell-only strategy. As a result, increasing transaction costs leads to an overall reduction of the proportion of predatory trades in equilibrium. In consequence,  both agents in our model can reduce their protection against predatory trading and therefore use more efficient strategies to carry out their trades. They can thus fully benefit from higher trading frequencies, which leads to the observed decrease of expected costs as a function of $N$ if $\theta$ is sufficiently large.  Moreover, for appropriate values of $\theta>0$, the benefit of increased efficiency outweighs the price to be paid in higher transaction costs and so an overall reduction of costs is achieved. 

Let us point out that, in the case $G(t)=e^{-\rho t}$,  many qualitative observations made in this section by means of numerical experiments have meanwhile been given rigorous mathematical proofs in our follow-up paper~\cite{SchiedStrehleZhang}, which has Elias Strehle as additional coauthor. There, we investigate the limits of equilibrium strategies and expected costs as $N\uparrow\infty$. We prove that, for $\theta=0$, both strategies and costs  oscillate indefinitely between two accumulation points, for which we provide explicit formulas. For $\theta>0$, however, strategies and  costs  converge toward limits that are independent of $\theta$. We then show that the limiting strategies form a Nash equilibrium for a continuous-time version of the model with $\theta=\theta^*$, and that the corresponding expected costs coincide with the high-frequency limits of the discrete-time equilibrium costs. For $\theta\neq\theta^*$, however, continuous-time Nash equilibria do not exist unless $X_0=Y_0=0$.

Another interesting question is the comparison of the expected costs of the equilibrium strategies with the expected costs that both agents would have if none of them were aware of the other's trading activities. In this case, a trader with initial inventory $Z_0$ will apply the strategy 
\begin{equation*}
\wh{\bm\zeta}^{Z_0}=\frac{Z_0}{\bm 1^\top\Gamma_\theta^{-1}\bm 1}\Gamma_\theta^{-1}\bm 1,
\end{equation*}
which is  the  strategy for a single trader facing the positive definite decay  $G$ and transaction costs measured by the parameter $\theta\ge0$; this follows by taking the positive definite decay kernel $G(t)+\Ind{\{0\}}(t)$  in~\cite[Proposition 1]{ASS}. We can thus define a \emph{price of anarchy} in our situation by letting
\begin{equation*}
\text{PoA}_N(\theta,X_0,Y_0):=\frac{\bE[\,\cC_{\bT_{N}}(\wh{\bm\zeta}^{X_0}|\wh{\bm\zeta}^{Y_0})\,]+\bE[\,\cC_{\bT_{N}}(\wh{\bm\zeta}^{Y_0}|\wh{\bm\zeta}^{X_0})\,]}{\bE[\,\cC_{\bT_{N}}(\bm\xi^*|\bm\eta^*)\,]+\bE[\,\cC_{\bT_{N}}(\bm\eta^*|\bm\xi^*)\,]},
\end{equation*}
where $\bm\xi^*$ and $\bm\eta^*$ are the equilibrium strategies from \eqref{equilibrium strategies}. See Figure~\ref{PoA figure}
 for a plot.

\begin{figure}
\hspace{1.5cm}
\begin{overpic}[width=10cm]{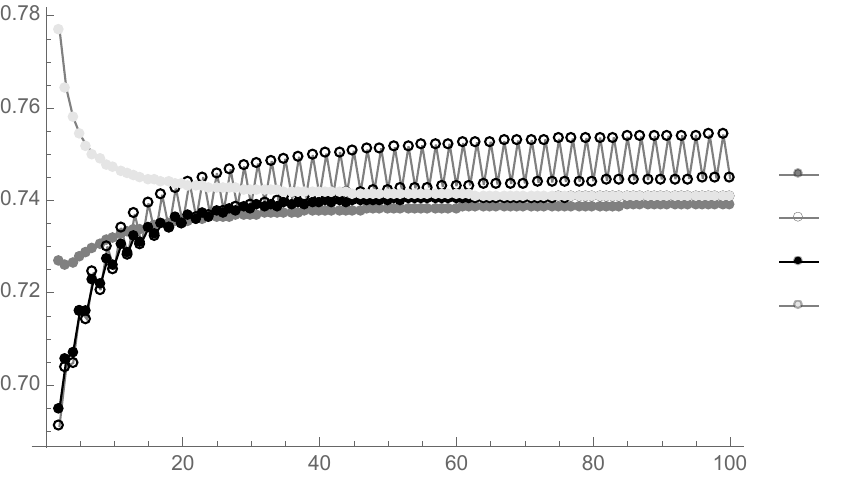}
\put(98,35.5){\small$\theta=0.1$}
\put(98,30.5){\small$\theta=0$}
\put(98,25.4){\small$\theta=0.01$}
\put(98,20.0){\small$\theta=\theta^*=0.25$}
\put(89,0){\small$N$}
\end{overpic}
\caption{Expected costs $\bE[\,\cC_{\bT_{N}}(\bm\xi^*|\bm\eta^*)\,]= \bE[\,\cC_{\bT_{N}}(\bm\eta^*|\bm\xi^*)\,]$ for various values of $\theta$ as a function of trading frequency, $N$, with the equidistant time grid $\bT_N$, $T=1$, $G(t)=e^{-t}$,  and $X_0=Y_0=1$.}\label{Monotonicity Cost fig}
\end{figure}

\begin{figure}
\centering
\begin{overpic}[width=10cm]{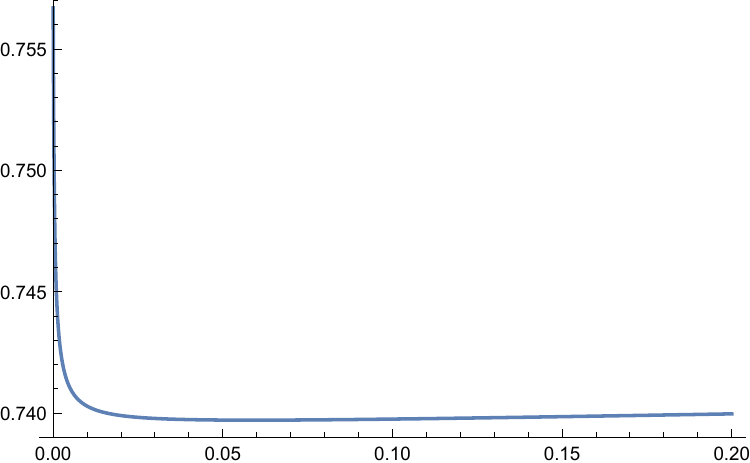}
\put(101,0){\small$\theta$}
\end{overpic}
\caption{ Expected costs $\bE[\,\cC_{\bT_{501}}(\bm\xi^*|\bm\eta^*)\,]$  as a function of $\theta$ for  initial values $X_0=Y_0=1$ and $G(t)=e^{-t}$. The costs decrease steeply from the value 0.7567 at $\theta=0$ until a minimum value of about $0.7397$ at $\theta=0.06$. From then on there is a moderate and almost linear increase with, e.g., a value of $0.7407$ at $\theta=0.5$. This increase corresponds to a slope of approximately 0.002. We took the  equidistant time grid $\bT_{501}$ and $\rho=1$.}
\label{costs(theta) fig}
\end{figure}

\begin{figure}
\centering
\begin{minipage}[b]{8.0cm}
\begin{overpic}[width=8cm]{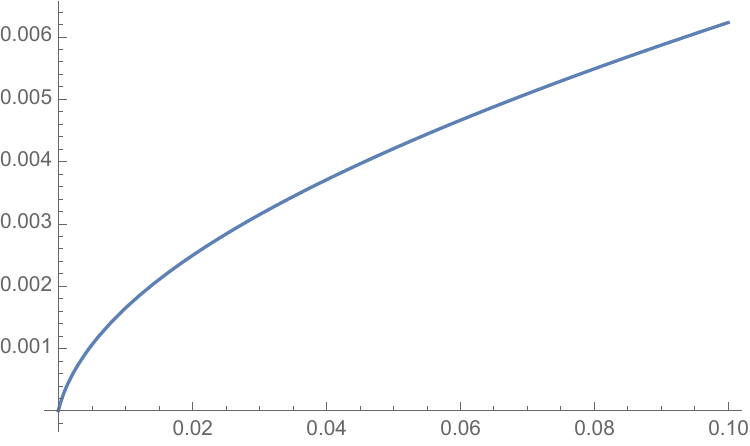}
\end{overpic}
\end{minipage}\qquad
\begin{minipage}[b]{8.0cm}
\begin{overpic}[width=8cm]{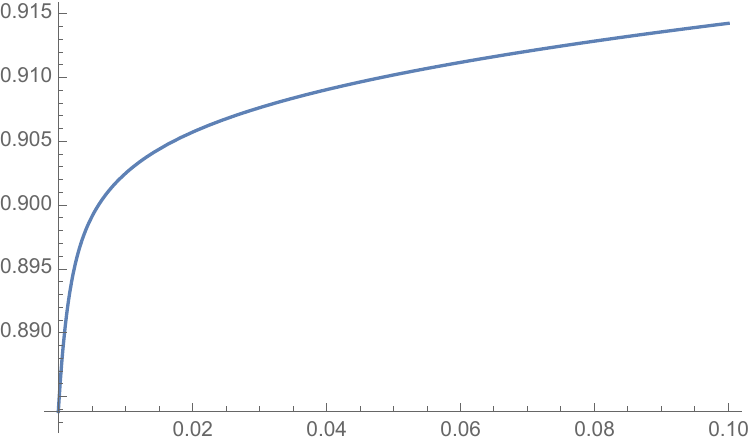}
\end{overpic}
\end{minipage}
\caption{Price of anarchy, $\text{PoA}_{101}(\theta,X_0,Y_0)$, as a function of $\theta$ for $X_0=1$ and $Y_0=-1$ (left) and $X_0=1$ and $Y_0=1$ (right)  for $G(t)=e^{-t}$. The steep increase on the right-hand panel is due to the initial decrease of the expected costs for $X_0=Y_0=1$ as shown in Figure~\ref{costs(theta) fig}. The steep of the price of anarchy in the right-hand panel is the result of the decrease of the corresponding equilibrium strategies as shown in Figure~\ref{costs(theta) fig}.}\label{PoA figure}
\end{figure}

\section{Proofs}\label{Proof Section}

\subsection{Proof of Theorem~\ref{th1} and Proposition~\ref{linear transaction costs prop}}

\begin{lemma}\label{lm1}
The expected costs of an admissible strategy $\bm \xi\in\mathscr{X}(X_0,\mathbb{T})$ given another admissible strategy $\bm \eta\in\mathscr{X}(Y_0,\mathbb{T})$ are
\begin{eqnarray}\label{x0s0}
\bE[\,\mathscr{C}_\bT(\bm \xi|\bm \eta)\,]&=&\bE\Big[\,\frac{1}{2}\bm \xi^\top\Gamma_\theta\bm \xi+\bm \xi^\top  \wt{{{\Gamma}}}\bm \eta\,\Big].
\end{eqnarray}
\end{lemma}

\begin{proof} Without loss of generality, we may assume $G(0)=1$. Since the sequence $(\eps_i)_{i=0,1,\dots}$ is independent of $\sigma(\bigcup_{t\ge0}\cF_t)$ and the two strategies $\bm \xi$ and $\bm \eta$ are measurable with respect to this $\sigma$-field, we get $\bE[\,\eps_k \xi_k\eta_k\,]=\frac1 2\bE[\,\xi_k\eta_k\,]$. Hence, 
\begin{eqnarray*}
\bE[\,\mathscr{C}_\bT(\bm \xi|\bm \eta)\,]-X_0S^0_0&=&\bE\bigg[\,\sum_{k=0}^N\Big(\frac12\xi_k^2-S_{t_k}^{\bm \xi,\bm \eta}\xi_k+\eps_k\xi_k\eta_k+\theta\xi_k^2\Big)\,\bigg]\\
&=&\bE\bigg[\,\sum_{k=0}^N\bigg(\frac12\xi_k^2+\frac12\xi_k\eta_k-\xi_k\Big(S^0_{t_k}- \sum_{m=0}^{k-1}(\xi_{m}+\eta_{m})G(t_k-t_m)\Big)+\theta\xi_k^2\bigg)\,\bigg]\\
&=&\bE\bigg[\,-\sum_{k=0}^N\xi_kS^0_{t_k}+\frac12\sum_{k=0}^N\xi_k^2+ \sum_{k=0}^N\xi_k\sum_{m=0}^{k-1}\xi_mG(t_k-t_m)\\
&&\qquad+\sum_{k=0}^N\bigg(\xi_k\Big(\frac12\eta_k+ \sum_{m=0}^{k-1}\eta_{m}G(t_k-t_m)\Big)+\theta\xi_k^2\bigg)\,\bigg].
\end{eqnarray*}
Since each $\xi_k$ is $\cF_{t_k}$-measurable and $S^0$ is a martingale, we get from condition (b) in Definition~\ref{Strategies def} that
$$\bE\Big[\,\sum_{k=0}^N\xi_kS^0_{t_k}\,\Big]=\bE\Big[\,\sum_{k=0}^N\xi_kS^0_{T}\,\Big]=X_0\bE[\,S^0_T\,]=X_0S_0^0.
$$
Moreover,
$$\frac12\sum_{k=0}^N\xi_k^2+ \sum_{k=0}^N\xi_k\sum_{m=0}^{k-1}\xi_mG(t_k-t_m)=\frac12\sum_{k,m=0}^N\xi_k\xi_m  G(|t_k-t_m|)=\frac12\bm \xi^\top  {\Gamma}\bm \xi,
$$
and
$$\sum_{k=0}^N\xi_k\Big(\frac12\eta_k+ \sum_{m=0}^{k-1}\eta_{m}G(t_k-t_m)\Big)=\bm \xi^\top  \wt {\Gamma}\bm \eta.
$$
Putting everything together yields the assertion.
\end{proof}

We will use the convention of saying that an $n\times n$-matrix $A$ is \emph{positive}  if $\bm x^\top A\bm x>0$ for all nonzero $\bm x\in\bR^{n}$, which makes sense also if $A$ is not necessarily symmetric. Clearly, for a positive   matrix $A$ there is no  nonzero $\bm x\in\bR^{n}$ for which $A\bm x=\bm0$, and so $A$ is invertible. Moreover, 
 writing a given nonzero $\bm x\in\bR^{n}$ as $\bm x=A\bm y$ for $\bm y=A^{-1}\bm x\neq\bm 0$, we see that $\bm x^\top A^{-1}\bm x=\bm y^\top A^\top\bm y=\bm y^\top A\bm y>0$. So the inverse of a positive    matrix is also positive. Recall that we assume \eqref{G strictly positive definite} throughout this paper.

\begin{lemma}\label{positive definite lemma}The matrices ${\Gamma_\theta}$, $\wt {\Gamma}$,  ${\Gamma_\theta}+\wt {\Gamma}$, ${\Gamma_\theta}-\wt {\Gamma}$   are positive     for all $\theta\ge0$. In particular, all terms in \eqref{v and w def} are well-defined and the denominators in \eqref{v and w def} are strictly positive.
\end{lemma}

\begin{proof}That ${\Gamma}$ is positive definite,  and hence positive,  follows directly from \eqref{G strictly positive definite}. Therefore, for nonzero ${\bm x}\in\bR^{N+1}$,
$$0<{\bm x}^\top {\Gamma}{\bm x}={\bm x}^\top(\wt {\Gamma}+\wt {\Gamma}^\top){\bm x}={\bm x}^\top \wt {\Gamma}{\bm x}+{\bm x}^\top\wt {\Gamma}^\top {\bm x}=2{\bm x}^\top \wt {\Gamma}{\bm x},
$$
which shows that the matrix $\wt {\Gamma}$ is positive. Next, ${\Gamma}-\wt {\Gamma}=\wt {\Gamma}^\top$ and so this matrix is also positive. Clearly, the sum of two positive   matrices is also positive, which  shows that  ${\Gamma_\theta}+\wt {\Gamma}={\Gamma}+\wt {\Gamma}+2\theta\Id$ and ${\Gamma_\theta}-\wt {\Gamma}={\Gamma}-\wt {\Gamma}+2\theta\Id$ are positive   for $\theta\ge0$.
\end{proof}

\begin{lemma}\label{uniqueness lemma} For given time grid $\bT$ and initial values $X_0$ and $Y_0$,
there exists at most one Nash equilibrium in the class $\cX(X_0,\bT)\times\cX(Y_0,\bT)$.
\end{lemma}

\begin{proof}
We assume by way of contradiction that there exist two distinct Nash equilibria $(\bm \xi^0,\bm \eta^0)$ and $(\bm \xi^1,\bm \eta^1)$ in $\cX(X_0,\bT)\times\cX(Y_0,\bT)$. Here, the fact that the two Nash equilibria are distinct means that they are not $\bP$-a.s.~equal. Then we define for $\alpha\in[0,1]$
$$
\bm \xi^{\alpha}:=\alpha\bm \xi^1+(1-\alpha)\bm \xi^0\qquad\text{and}\qquad \bm \eta^{\alpha}:=\alpha\bm \eta^1+(1-\alpha)\bm \eta^0.
$$
We furthermore let 
$$
f(\alpha):=\bE\Big[\,\cC_\bT(\bm \xi^{\alpha}|\bm \eta^0)+\cC_\bT(\bm \eta^{\alpha}|\bm \xi^0)+\cC_\bT(\bm \xi^{1-\alpha}|\bm \eta^1)+\cC_\bT(\bm \eta^{1-\alpha}|\bm \xi^1)\,\Big].
$$
Since according to \eqref{G strictly positive definite}
 the matrix $ \Gamma_\theta$ is  positive definite, the  functional
$$
\bm \xi\longmapsto\bE[\,\cC_\bT(\bm  \xi|\bm \eta)\,]=\bE\Big[\,\frac{1}{2}\bm \xi^\top\Gamma_\theta\bm \xi+\bm \xi^\top  \wt{{{\Gamma}}}\bm \eta\,\Big]
$$
is strictly convex with respect to $\bm \xi$. Since the two Nash equilibria $(\bm \xi^0,\bm \eta^0)$ and $(\bm \xi^1,\bm \eta^1)$ are distinct, $f(\alpha)$ must also be strictly convex in $\alpha$ and have its unique minimum in $\alpha=0$. That is,
\begin{equation*}
f(\alpha)>f(0)\quad\mbox{for }\alpha>0.
\end{equation*}
It follows that
\begin{equation}\label{fge0}
\lim_{h\downarrow0}\frac{f(h)-f(0)}{h}=\frac{df(\alpha)}{d\alpha}\Big|_{\alpha=0+}\ge0.
\end{equation}
Next, by the symmetry of $\Gamma_\theta$,
\begin{eqnarray*}\bE[\,\cC_\bT(\bm \xi^\alpha|\bm \eta)\,]&=&\bE\bigg[\,\frac{1}{2}\alpha^2(\bm \xi^1)^\top \Gamma_\theta\bm \xi^1+\alpha(1-\alpha)(\bm \xi^1)^\top \Gamma_\theta\bm \xi^0+\frac{1}{2}(1-\alpha)^2(\bm \xi^0)^\top \Gamma_\theta\bm \xi^0\\
&&\qquad\qquad\qquad+\alpha(\bm \xi^1)^\top \wt{{\Gamma}}\bm \eta+(1-\alpha)(\bm \xi^0)^\top \wt{{\Gamma}}\bm \eta\,\bigg].
\end{eqnarray*}
Therefore,
$$
\frac{d}{d\alpha}\Big|_{\alpha=0+}\bE[\,\cC_\bT(\bm \xi^\alpha|\bm \eta)\,]=\bE\Big[\,(\bm \xi^1-\bm \xi^0)^\top \Gamma_\theta\bm \xi^0+(\bm \xi^1-\bm \xi^0)^\top \wt{{\Gamma}}\bm \eta\,\Big].$$
Hence, it follows that
\begin{eqnarray*}
\lefteqn{\frac{d}{d\alpha}\Big|_{\alpha=0+}f(\alpha)}\\
&=&\bE\bigg[(\bm \xi^1-\bm \xi^0)^\top{\Gamma_\theta}\bm \xi^0+(\bm \xi^1-\bm \xi^0)^\top \wt{{\Gamma}}\bm \eta^0+(\bm \xi^0-\bm \xi^1)^\top{\Gamma_\theta}\bm \xi^1+(\bm \xi^0-\bm \xi^1)^\top \wt{{\Gamma}}\bm \eta^1\\
&&\qquad +(\bm \eta^1-\bm \eta^0)^\top{\Gamma_\theta}\bm \eta^0+(\bm \eta^1-\bm \eta^0)^\top \wt{{\Gamma}}\bm \xi^0+(\bm \eta^0-\bm \eta^1)^\top{\Gamma_\theta}\bm \eta^1+(\bm \eta^0-\bm \eta^1)^\top \wt{{\Gamma}}\bm \xi^1\bigg]\\
&=&-\bE\bigg[(\bm \xi^1-\bm \xi^0)^\top{\Gamma_\theta}(\bm \xi^1-\bm \xi^0)+(\bm \eta^1-\bm \eta^0)^\top{\Gamma_\theta}(\bm \eta^1-\bm \eta^0)\bigg]\\
&&+\bE\bigg[(\bm \xi^1-\bm \xi^0)^\top \wt{{\Gamma}}(\bm \eta^0-\bm \eta^1)+(\bm \xi^1-\bm \xi^0)^\top \wt{{\Gamma}}^\top(\bm \eta^0-\bm \eta^1)\bigg]\\
&=&-\bE\bigg[(\bm \xi^1-\bm \xi^0)^\top{\Gamma_\theta}(\bm \xi^1-\bm \xi^0)+(\bm \eta^1-\bm \eta^0)^\top{\Gamma_\theta}(\bm \eta^1-\bm \eta^0)\bigg]
-\bE\bigg[(\bm \xi^1-\bm \xi^0)^\top   {\Gamma}(\bm \eta^1-\bm \eta^0)\bigg].
\end{eqnarray*}
Now,
\begin{eqnarray*}
\lefteqn{(\bm \xi^1-\bm \xi^0)^\top{  {\Gamma}}(\bm \eta^1-\bm \eta^0)+\frac{1}{2}\Big((\bm \xi^1-\bm \xi^0)^\top{\Gamma_\theta}(\bm \xi^1-\bm \xi^0)+(\bm \eta^1-\bm \eta^0)^\top{\Gamma_\theta}(\bm \eta^1-\bm \eta^0)\Big)}\\
&&\qquad\ge\frac12\Big((\bm \xi^1-\bm \xi^0+\bm \eta^1-\bm \eta^0)^\top  {\Gamma}(\bm \xi^1-\bm \xi^0+\bm \eta^1-\bm \eta^0)\Big)\ge0.
\end{eqnarray*}
Thus, and because the two  Nash equilibria $(\bm \xi^0,\bm \eta^0)$ and $(\bm \xi^1,\bm \eta^1)$ are distinct, we have
$$
\frac{d}{d\alpha}\Big|_{\alpha=0+}f(\alpha)\le-\frac12\bE\bigg[(\bm \xi^1-\bm \xi^0)^\top  {\Gamma}(\bm \xi^1-\bm \xi^0)+(\bm \eta^1-\bm \eta^0)^\top  {\Gamma}(\bm \eta^1-\bm \eta^0)\bigg]<0,
$$
which contradicts \eqref{fge0}. Therefore, there can exist at most one Nash equilibrium in the class $\cX(X_0,\bT)\times\cX(Y_0,\bT)$.
\end{proof}

Now let us introduce the class
$$\cX_{\text{\rm det}}(Z_0,\bT):=\Big\{\bm \zeta\in \cX(Z_0,\bT)\,\Big|\,\text{$\bm \zeta$ is deterministic}\Big\}
$$
of deterministic strategies in $\cX(Z_0,\bT)$. A Nash equilibrium in the class $\cX_{\text{\rm det}}(X_0,\bT)\times\cX_{\text{\rm det}}(Y_0,\bT)$ is defined in the same way as in Definition~\ref{Nash Def}. 

\begin{lemma}\label{deterministic lemma}
A Nash equilibrium in the class $\cX_{\text{\rm det}}(X_0,\bT)\times\cX_{\text{\rm det}}(Y_0,\bT)$ of deterministic strategies is also a Nash equilibrium in the class $\cX(X_0,\bT)\times\cX(Y_0,\bT)$ of adapted strategies.
\end{lemma}

\begin{proof}
Assume that $(\bm \xi^*,\bm \eta^*)$ is a Nash equilibrium in the class $\cX_{\text{\rm det}}(X_0,\bT)\times\cX_{\text{\rm det}}(Y_0,\bT)$ of deterministic strategies. We need to show that $\bm \xi^*$ minimizes $\bE[\,\cC_\bT(\bm \xi|\bm \eta^*)\,]$  and $\bm \eta^*$ minimizes $\bE[\,\cC_\bT(\bm \eta|\bm \xi^*)\,]$ in the respective classes $\cX(X_0,\bT)$ and $\cX(Y_0,\bT)$ of adapted strategies. To this end, let $\bm \xi\in \cX(X_0,\bT)$ be given. We define  $\bbar{\bm \xi}\in \cX_{\text{\rm det}}(X_0,\bT)$ by $\bbar{\xi}_k=\bE[\,\xi_k\,]$ for $k=0,1,\dots,N$.

Applying  Jensen's inequality to the convex function $\bR^{N+1}\ni {\bm x}\mapsto  {\bm x}^\top \Gamma_\theta {\bm x}$, we obtain 
\begin{eqnarray*}
\bE[\,\cC_\bT(\bm \xi|\bm \eta^*)\,]&=&\bE\Big[\frac12\bm \xi^\top \Gamma_\theta\bm \xi+\bm \xi^\top  \wt{{\Gamma}}\bm \eta^*\Big]=\bE\Big[\frac12\bm \xi^\top \Gamma_\theta\bm \xi\Big]+\bbar{\bm \xi}^\top  \wt{{\Gamma}}\bm \eta^*\\
&\ge&\frac12\bbar{\bm \xi}^\top \Gamma_\theta\bbar{\bm \xi}+\bbar{\bm \xi}^\top  \wt{{\Gamma}}\bm \eta^*=\bE[\,\cC_\bT(\bbar{\bm \xi}|\bm \eta^*)\,]\\
&\ge&\bE[\,\cC_\bT( \bm \xi^*|\bm \eta^*)\,].
\end{eqnarray*}
This shows that $\bm \xi^*$ minimizes $\bE[\,\cC_\bT(\bm \xi|\bm \eta^*)\,]$ over $\bm \xi\in \cX(X_0,\bT)$. One can show analogously that $\bm \eta^*$ minimizes $\bE[\,\cC_\bT(\bm \eta|\bm \xi^*)\,]$  over $\bm \eta\in\cX(Y_0,\bT)$, which completes the proof.
\end{proof}

\begin{remark}Before proving Theorem~\ref{th1}, we briefly explain how to derive heuristically the explicit form \eqref{equilibrium strategies}  of the equilibrium strategies. By Lemma~\ref{lm1} and the method of Lagrange multipliers, a necessary condition for  $( {\bm \xi}^*, {\bm \eta}^*)$ to be a Nash equilibrium in $\cX_{\text{\rm det}}(X_0,\bT)\times\cX_{\text{\rm det}}(Y_0,\bT)$  is the existence of  $\alpha,\beta\in\bR$, such that
\begin{equation}\label{th1eqn}
\left\{
\begin{aligned}
&\Gamma_\theta {\bm \xi}^*+ \wt{{{\Gamma}}} {\bm \eta}^*=\alpha{\bm 1};\\
&\Gamma_\theta {\bm \eta}^*+ \wt{{{\Gamma}}} {\bm \xi}^*=\beta{\bm 1}.
\end{aligned}
\right.
\end{equation}
By adding the  equations in \eqref{th1eqn} we obtain
\begin{equation}\label{les}
(\Gamma_\theta+ \wt {\Gamma})( {\bm \xi}^*+ {\bm \eta}^*)=(\alpha+\beta){\bm 1}.
\end{equation}
By Lemma~\ref{positive definite lemma}, the matrix $\Gamma_\theta+ \wt {\Gamma}$ is positive   and hence invertible, so that \eqref{les} can be solved for $ {\bm \xi}^*+ {\bm \eta}^*$.  Since we must also have ${\bm 1}^\top( {\bm \xi}^*+ {\bm \eta}^*)=X_0+Y_0$, we obtain 
$$
 {\bm \xi}^*+ {\bm \eta}^*=\frac{(X_0+Y_0)}{{\bm 1}^\top (\Gamma_\theta+ \wt {\Gamma})^{-1}{\bm 1}}(\Gamma_\theta+ \wt {\Gamma})^{-1}{\bm 1}=(X_0+Y_0)\bm v.
$$
Similarly, by subtracting the  two equations in  \eqref{th1eqn} yields 
$$
(\Gamma_\theta- \wt{{{\Gamma}}})( {\bm \xi}^*- {\bm \eta}^*)=(\alpha-\beta){\bm 1}.
$$
It follows again from Lemma~\ref{positive definite lemma} that $(\Gamma_\theta- \wt{{{\Gamma}}})$ is invertible, and so we have 
$$
 {\bm \xi}^*- {\bm \eta}^*=\frac{(X_0-Y_0)}{{\bm 1}^T(\Gamma_\theta- \wt{{{\Gamma}}})^{-1}{\bm 1}}(\Gamma_\theta- \wt{{{\Gamma}}})^{-1}{\bm 1}=(X_0-Y_0)\bm w.
$$
Thus, $ {\bm \xi}^*$ and $ {\bm \eta}^*$ ought to  be given by \eqref{equilibrium strategies}.
\end{remark}

\bigskip

\begin{proof}[Proof of Theorem~\ref{th1}] By Lemmas~\ref{uniqueness lemma} and~\ref{deterministic lemma} all we need to show is that \eqref{equilibrium strategies} defines a Nash equilibrium in the class $\cX_{\text{\rm det}}(X_0,\bT)\times\cX_{\text{\rm det}}(Y_0,\bT)$ of deterministic strategies. For $(\bm \xi,\bm \eta)\in \cX_{\text{\rm det}}(X_0,\bT)\times\cX_{\text{\rm det}}(Y_0,\bT)$ we have
\begin{equation}\label{det exp cost}
\bE[\,\cC_\bT(\bm \xi|\bm \eta)\,]=\frac12\bm \xi^\top \Gamma_\theta\bm \xi+\bm \xi^\top \wt {\Gamma}\bm \eta.
\end{equation}
 Therefore minimizing $\bE[\,\cC_\bT(\bm \xi|\bm \eta)\,]$ over $\bm \xi\in \cX_{\text{\rm det}}(X_0,\bT)$ is equivalent to the minimization of the quadratic form on the right-hand side of \eqref{det exp cost} over $\bm \xi\in\bR^{N+1}$ under the constraint ${\bm 1}^\top \bm \xi=X_0$.

Now we  prove that the strategies $\bm \xi^*$ and $\bm \eta^*$ given by \eqref{equilibrium strategies} are indeed optimal.  We have
\begin{equation*}
\Gamma_\theta\bm \xi^*+ \wt {\Gamma}\bm \eta^*=\frac12(X_0+Y_0)(\Gamma_\theta+ \wt {\Gamma})\bm v+\frac12(X_0-Y_0)(\Gamma_\theta- \wt {\Gamma})\bm w=\mu\bm 1,
\end{equation*}
where 
$$\mu=\frac{(X_0+Y_0)}{2\bm 1^\top (\Gamma_\theta+ \wt {\Gamma})\bm 1}+\frac{(X_0-Y_0)}{2\bm 1^\top (\Gamma_\theta- \wt {\Gamma})\bm 1}.
$$
Now let $\bm \xi\in \cX_{\text{\rm det}}(X_0,\bT)$ be arbitrary and define $\bm \zeta:=\bm \xi-\bm \xi^*$. Then we have $\bm \zeta^\top\bm1=0$. Hence,  by the symmetry of $\Gamma_\theta$,
\begin{eqnarray*}
\frac12\bm \xi^\top \Gamma_\theta\bm \xi+\bm \xi^\top \wt {\Gamma}\bm \eta^*&=&\frac12(\bm \xi^*)^\top \Gamma_\theta\bm \xi^*+\frac12\bm \zeta^\top \Gamma_\theta\bm \zeta+\bm \zeta^\top \Gamma_\theta\bm \xi^*+(\bm \xi^*)^\top \wt {\Gamma}\bm \eta^*+\bm \zeta^\top \wt {\Gamma}\bm \eta^*\\
&=&\frac12(\bm \xi^*)^\top \Gamma_\theta\bm \xi^*+(\bm \xi^*)^\top \wt {\Gamma}\bm \eta^*+\frac12\bm \zeta^\top \Gamma_\theta\bm \zeta+\mu\bm \zeta^\top\bm1\\
&\ge&\frac12(\bm \xi^*)^\top \Gamma_\theta\bm \xi^*+(\bm \xi^*)^\top \wt {\Gamma}\bm \eta^*,
\end{eqnarray*}
where in the last step we have used that $\Gamma_\theta$ is positive definite and that $\bm \zeta^\top\bm1=0$. Therefore $\bm \xi^*$ minimizes \eqref{det exp cost} in the class $\cX_{\text{\rm det}}(X_0,\bT)$ for $\bm \eta=\bm \eta^*$. In the same way, one shows that $\bm \eta^*$ minimizes $\bE[\,\cC_\bT(\bm \eta|\bm \xi^*)\,]$ over $\bm \eta\in\cX_{\text{\rm det}}(X_0,\bT)$.
\end{proof}

\begin{proof}[Proof of Proposition~\ref{linear transaction costs prop}]  Following Lemma~\ref{lm1}, the expected cost functional with $x\mapsto\tau(|x|)$ replacing $x\mapsto\theta x^2$ is given by
$$\bE[\,\bbar\cC_{\bT}(\bm\xi|\bm \eta)\,]:=\bE\bigg[\,\bm\xi^\top\Gamma\bm\xi+\bm\xi^\top\wt\Gamma\bm\eta+\sum_{k=0}^N\tau(|\xi_k|)\,\bigg], \quad\bm\xi\in\cX(X_0,\bT),\ \bm\eta\in\cX(Y_0,\bT).
$$

Now let $\bm\xi^*$ and $\bm\eta^*$ be as in Theorem~\ref{th1}. Since both $\bm\xi^*$ and $\bm\eta^*$ are deterministic,  $|\xi^*_k|$ and $|\eta_k^*|$ take  just finitely many values as $k$ ranges from 0 to $N$. After adding the value 0 to this list and arranging it in increasing order, the values from that list correspond to numbers  $0=c_0< c_1<c_2<\cdots <c_{M-1}$. Then we take $c_{M}:=c_{M-1}+1$ and let $\tau:[0,\infty)\to [0,\infty)$ be the linear interpolation of the function $ x\mapsto\theta x^2$  with respect to the grid $c_0, c_1,\dots, c_{M}$ and with linear continuation beyond $[c_{M-1},c_{M}]$. Then $\tau(|\xi_k^*|)=\theta(\xi_k^*)^2$ and $\tau(|\eta_k^*|)=\theta(\eta_k^*)^2$ holds for all $k$, and it follows that  $\bE[\,\bbar\cC_{\bT}(\bm\xi^*|\bm \eta^*)\,]=\bE[\,\cC_{\bT}(\bm\xi^*|\bm \eta^*)\,]$. 
 
 Let us now suppose by way of contradiction that $(\bm\xi^*,\bm\eta^*)$ is not a Nash equilibrium in $\cX(X_0,\bT)\times \cX(Y_0,\bT)$. Then there  exist $\bm\xi\in \cX(X_0,\bT)$ or $\bm\eta\in \cX(Y_0,\bT)$ such that $\bE[\,\bbar\cC_{\bT}(\bm\xi|\bm \eta^*)\,]<\bE[\,\bbar\cC_{\bT}(\bm\xi^*|\bm \eta^*)\,]$ or $\bE[\,\bbar\cC_{\bT}(\bm\eta|\bm \xi^*)\,]<\bE[\,\bbar\cC_{\bT}(\bm\eta^*|\bm \xi^*)\,]$. By symmetry, it is sufficient to consider only the first possibility. For $\alpha\in[0,1]$, let $\bm\xi^\alpha:=(1-\alpha)\bm\xi^*+\alpha\bm\xi$. By the convexity of the expected cost functional, we have $\bE[\,\bbar\cC_{\bT}(\bm\xi^\alpha|\bm \eta^*)\,]<\bE[\,\bbar\cC_{\bT}(\bm\xi^*|\bm \eta^*)\,]$ for all $\alpha\in(0,1]$. By using the boundedness of admissible strategies (Definition~\ref{Strategies def} (a)), there is $\eps\in(0,1]$ such  that 
$|\xi^\eps_k|\le c_{M}$  for  $k=0,\dots, N$ $\bP$-a.s. Thus, the convexity of $x\mapsto\theta x^2$ implies that $\tau(|\xi^\eps_k|)\ge\theta(\xi^\eps_k)^2$ $\bP$-a.s. for $k=0,\dots, N$. Hence, 
$$\bE[\,\bbar\cC_{\bT}(\bm\xi^\eps|\bm \eta^*)\,]\ge \bE[\,\cC_{\bT}(\bm\xi^\eps|\bm \eta^*)\,]>\bE[\,\cC_{\bT}(\bm\xi^*|\bm \eta^*)\,]=\bE[\,\bbar\cC_{\bT}(\bm\xi^*|\bm \eta^*)\,],
$$
which is the desired contradiction. 
\end{proof}

\subsection{Proof of Propositions~\ref{three oscillations prop} and~\ref{oscillations prop}}

\begin{proof}[Proof of Proposition~\ref{three oscillations prop}] According to \eqref{v and w def} and Lemma~\ref{positive definite lemma}, the vector $\bm w$ is a positive multiple of $(\Gamma_\theta-\wt \Gamma)^{-1}\bm1$. The matrix $\Gamma_\theta-\wt \Gamma$ is an invertible upper triangular matrix, whose diagonal entries are all equal to $\nu:=G(0)/2+2\theta$. We may assume without loss of generality that $\nu=1$; otherwise we divide $G$ by $\nu$. Then we will have $G(0)>1$, and there exists $\delta_1>0$ such that also $G(\delta_1)>1$. Now we take $\delta\le\delta_1$ such that $G$ is nonincreasing in $[0,2\delta]$. The off-diagonal elements of $\Gamma_\theta-\wt \Gamma$ are equal to  $\Gamma_{i,j}=G(t_{j-1}-t_{i-1})$ for $i<j$ and they vanish for $i>j$. Let $\bm u=(u_1,\dots, u_{N+1})^\top=(\Gamma_\theta-\wt \Gamma)^{-1}\bm1$. A straightforward computation shows that 
$$u_{N+1}=1,\quad u_N=1-\Gamma_{N,N+1},\quad\text{and}\quad u_{N-1}=1-\Gamma_{N-1,N+1}+\Gamma_{N-1,N}(\Gamma_{N,N+1}-1).
$$
Clearly, $u_{N+1}>0$ holds trivially. Next, due to our choice of $\delta$, we have $\Gamma_{i-1,i}=G(t_i-t_{i-1})>1$ for $i=N,N+1$. In particular, $u_N<0$ follows. Moreover, 
$$u_{N-1}>1-\Gamma_{N-1,N+1}+(\Gamma_{N,N+1}-1)=G(t_N-t_{N-1})-G(t_N-t_{N-2})\ge0,
$$
where the latter inequality follows from the assumption that $G$ is nonincreasing in $[0,2\delta]$.
\end{proof}

\noindent{\it Proof of Proposition~\ref{oscillations prop}.} Recall that here   $G(t)=\lambda e^{-\rho t}$ for constants $\lambda,\rho>0$. We need to compute the inverse of the matrix ${{\Gamma_\theta}}- \wt{{{\Gamma}}}$. Setting $\kappa:=2\theta/\lambda+\frac12$ and $a:=e^{-\rho T}$, we have
$$ {{\Gamma_\theta}}- \wt{{{\Gamma}}}=\lambda\begin{pmatrix}
  \,&\kappa&a^{\frac{1}{N}}&a^{\frac{2}{N}}&\cdots&a^{\frac{N-1}{N}}&a \,{} \\
     &0&\kappa&a^{\frac{1}{N}}&\cdots&a^{\frac{N-2}{N}}&a^{\frac{N-1}{N}} \,{} \\
     &0&\ddots&\ddots&\ddots&\ddots&\vdots\,{} \\
     &\vdots&\ddots&\ddots&\ddots&\ddots&\vdots\,{} \\
     &\vdots&\ddots&\ddots&\ddots&\kappa&a^{\frac{1}{N}}\,{} \\
     &0&\cdots&\cdots&\cdots&0&\kappa\,{}
 \end{pmatrix}.
$$
It is easy to verify that the inverse of this matrix is given by
\begin{equation*}
\Pi_N:=\frac1\lambda\begin{pmatrix}
  \,&\frac1\kappa&\frac{-a^{\frac{1}{N}}}{\kappa^2}&\frac{-a^{\frac{2}{N}}(\kappa-1)}{\kappa^3}&\cdots&\frac{-a^{\frac{N-1}{N}}(\kappa-1)^{N-2}}{\kappa^N}&\frac{-a^{\frac{N}{N}}(\kappa-1)^{N-1}}{\kappa^{N+1}} \,{} \\
     &0&\frac1\kappa&\frac{-a^{\frac{1}{N}}}{\kappa^2}&\cdots&\frac{-a^{\frac{N-2}{N}}(\kappa-1)^{N-3}}{\kappa^{N-1}}&\frac{-a^{\frac{N-1}{N}}(\kappa-1)^{N-2}}{\kappa^N}\,{} \\
     &0&\ddots&\ddots&\ddots&\ddots&\vdots\,{} \\
     &\vdots&\ddots&\ddots&\ddots&\ddots&\vdots\,{} \\
     &\vdots&\ddots&\ddots&\ddots&\frac1\kappa&\frac{-a^{\frac{1}{N}}}{\kappa^2}\,{} \\
     &0&\cdots&\cdots&\cdots&0&\frac1\kappa\,{}
 \end{pmatrix}.
\end{equation*}
Let us denote by ${\bm u}=(u_1,u_2,\dots,u_{N+1})\in\bR^{N+1}$ the vector $\lambda\Pi_N\bm1$. Then we have $u_{N+1}=\frac1\kappa$ and, for $n=1,\dots, N$, $u_n=u_{n+1}-a^{(N+1-n)/N}(\kappa-1)^{N-n}/\kappa^{N+2-n}$. That is, 
\begin{equation}\label{un eq}
\begin{split}
u_n&=\frac1\kappa-\frac{a^{\frac1N}}{\kappa^2}\sum_{m=n}^N\Big(\frac{a^{\frac1N}(\kappa-1)}{\kappa}\Big)^{N-m}=\frac1\kappa-\frac{a^{\frac1N}}{\kappa^2}\sum_{k=0}^{N-n}\Big(\frac{a^{\frac1N}(\kappa-1)}{\kappa}\Big)^{k}\\
&=\frac1\kappa\bigg[1-\frac{a^{\frac1N}}{\kappa(1-a^{\frac1N})+a^{\frac1N}}+(-1)^{N+1-n}\frac{a^{\frac1N}}{\kappa(1-a^{\frac1N})+a^{\frac1N}}\Big(\frac{a^{\frac1N}(1-\kappa)}{\kappa}\Big)^{N+1-n}\bigg].
\end{split}
\end{equation}

 If $\theta=0$, we have 
\begin{equation*}
u_n=2\bigg[1-\frac{2a^{\frac1N}}{1+a^{\frac1N}}+(-1)^{N+1-n}\frac{{2a^{\frac{N+2-n}N}}}{1+a^{\frac1N}}\bigg].
\end{equation*}
Since $a<1$, we have
$$0\le1-\frac{2a^{\frac1N}}{1+a^{\frac1N}}<1-a^{\frac1N}\longrightarrow0\qquad\text{as $N\ua\infty$.}
$$
On the other hand, we have  
$$\frac{{2a^{\frac{N+2-n}N}}}{1+a^{\frac1N}}\ge a^{\frac{N+2-n}N}\ge a^{\frac{N+1}N}\longrightarrow a\qquad\text{as $N\ua\infty$.}
$$
Therefore, the signs of $u_n$ will alternate as soon as $N$ is large enough to have $1-a^{\frac1N}<a^{\frac{N+1}N}$. This proves part (a). As for part (b), since the expression \eqref{un eq} is continuous in $\kappa$, the signs of $u_n$ will still alternate if, for fixed $N\ge N_0$, we take $\kappa$ slightly larger than $1/2$. (Note however that the term $(1-\kappa)^N/\kappa^N$ tends to zero faster than $1-a^{\frac1N}$, so we cannot get this result uniformly in $N$). 
\qed

\subsection{Proof of Theorem~\ref{nonosc thm}}

\begin{proof}[Proof of (a)$\Rightarrow$(b) in Theorem~\ref{nonosc thm}.] It is well known and easy to see that $G(0)\ge G(t)$ for all $t\ge0$, due to our assumption that the function $G(|\cdot|)$ is positive definite. The log-convexity of $G$ therefore implies that $G$ must be nonincreasing in a neighborhood of zero. Therefore, Proposition~\ref{three oscillations prop} is applicable. It implies that $\bm w$ must have some components with negative sign if $\theta<\theta^*$. This yields the assertion.
\end{proof}

The proof of the implication (b)$\Rightarrow$(a) in Theorem~\ref{nonosc thm} relies on the following classical result on the signs of power series, which is due to Kaluza~\cite{Kaluza} and Szeg\H o~\cite{Szego}. Here we state it in the formulation of Jurkat~\cite[Theorem 3]{Jurkat}.

\begin{theorem}[Kaluza sign criterion]\label{Kaluza thm} For $n\ge0$, let $a_n>0$ be coefficients in the power series $f(x)=\sum_{n=0}^\infty a_nx^n$ satisfying the condition that $a_{n+1}/a_n$ is nondecreasing in $n\ge0$. Then the coefficients $b_n$ of the formal reciprocal power series
$$\frac1{f(x)}=\sum_{n=0}^\infty b_nx^n
$$
satisfy $b_0=1/a_0>0$ and $b_n\le0$ for $n\ge1$. If, moreover, the power series for $f$ is convergent for  $|x|<1$, then it follows that $\lim_{x\ua 1}\frac1{f(x)}=\sum_{n=0}^\infty b_n$ exists and is nonnegative.
\end{theorem}

This result is connected with our situation as follows. Let $(a_n)_{n\ge0}$ be a sequence of numbers such that $a_0>0$ and consider the upper triangular Toeplitz matrix $A=(\wt a_{i,j})_{i,j=1,\dots, N}$ with coefficients $\wt a_{i,j}=a_{j-i}$ if $i\le j$ and $\wt a_{i,j}=0$ otherwise. The inverse $B=A^{-1}$ is then also an upper triangular Toeplitz matrix. It is generated by the sequence $(b_n)_{n\ge0}$ that satisfies $b_0=1/a_0$ and is otherwise determined recursively through the convolution identities
\begin{equation*}
\sum_{k=0}^m a_kb_{m-k}=0,\qquad m\ge1.
\end{equation*}
But these conditions also determine the coefficients $(b_n)_{n\ge0}$  of the (formal) reciprocal of the power series $\sum_{n=0}^\infty a_nx^n$, so that there is a one-to-one correspondence between the inversion of triangular Toeplitz matrices and the formal development of  reciprocal power series; see~\cite{Trench3}. 

\begin{proof}[Proof of (b)$\Rightarrow$(a) in Theorem~\ref{nonosc thm}.] Let $a_0=G(0)/2+2\theta$ and $a_n=G(nT/N)$ for $n\ge1$. Then the matrix $\Gamma_\theta-\wt\Gamma$ is equal to the upper triangular Toeplitz matrix constructed as above from the sequence $(a_n)_{n\ge0}$. Clearly, we have $a_n>0$ for all $n$, and the fact that $G$ is log-convex implies that $a_{n+1}/a_n$ is nondecreasing in $n\ge1$. If $\theta\ge\theta^*$, then we will also have $a_1/a_0\le a_2/a_1$. Moreover, the fact that $G$ is positive definite implies once again that $G(t)\le G(0)$ for all $t$ so that the sequence $(a_n)_{n\in\bN}$ is bounded and the power series $\sum_{n=0}^\infty a_nx^n$ converges for $|x|<1$. It follows that we may apply all parts of Theorem~\ref{Kaluza thm}. It yields that the coefficients $(b_n)_{n\ge0}$ satisfy $b_0>0$, $b_n\le0$ for $n\ge1$, and that $\sum_{n=0}^\infty b_n$ exists and is nonnegative. Therefore, we must have $\sum_{n=0}^k b_n\ge 0$ for all $k\ge0$. But these sums coincide with the components of the vector $(\Gamma_\theta-\wt\Gamma)^{-1}\bm1$, which is in turn proportional to $\bm w$.
\end{proof}

\medskip

\begin{proof}[Proof of {\rm (c)}$\Rightarrow${\rm (b)} in Theorem~\ref{nonosc thm}]
We consider the case  $N=1$. By definition,  $\bm v$ is proportional to the vector
   $$2\det(\Gamma_\theta+\wt\Gamma)(\Gamma_\theta+\wt\Gamma)^{-1}\bm1=\left(
\begin{array}{c}
\lambda  (3 -2 a)+ \gamma +4 \theta\\
\lambda  (3-4a)- \gamma +4\theta\end{array}
\right). $$
Clearly, the first component of this vector is positive for all $a\in(0,1)$ and $\theta\ge0$. By sending $a\ua1$ one sees, however, that the second component is negative for  $\theta<\theta^*$ and  $a$ sufficiently close to 1. Thus, we cannot have $\bm v\ge0$ in this case.
\end{proof}

Now we prepare for the proof of the implication {\rm (b)}$\Rightarrow${\rm (c)} in Theorem~\ref{nonosc thm}.
It relies on results for so-called $M$-matrices stated in the book~\cite{BermanPlemmons} by Berman and Plemmons. We first introduce some notations. If $A$ is a matrix or vector, we will write
\begin{enumerate}
\item $A\ge0$ if each entry of $A$ is nonnegative;
\item $A>0$ if $A\ge0$ and at least one entry is strictly positive;
\item $A\gg0$ if each entry of $A$ is strictly positive.
\end{enumerate}

\begin{definition}[Definition 1.2 in Chapter 6 of~\cite{BermanPlemmons}]
A matrix $A\in\bR^{n\times n}$  is called  a nonsingular \emph{$M$-matrix} if it is of the form
$
A=s\Id-B$, where the matrix $B\in \bR^{n\times n}$ satisfies $B\ge0$ and the parameter $s>0$ is strictly larger than the  spectral radius of $B$.
\end{definition}

Also recall that a matrix $A\in\bR^{n\times n}$ is called  a \emph{$Z$-matrix} if all its off-diagonal elements are nonpositive. 
Berman and Plemmons~\cite{BermanPlemmons} give 50 equivalent characterizations of the fact that a given $Z$-matrix is  a nonsingular $M$-matrix. We will need three of them here and summarize them in the following statement.

\begin{theorem}[From Theorem 2.3  in Chapter 6 of~\cite{BermanPlemmons}]\label{mm}
For a $Z$-matrix $A\in\bR^{n\times n}$,  the following conditions are equivalent.
\begin{enumerate}
\item $A$ is a nonsingular $M$-matrix;
\item All the leading principal minors of $A$ are positive.
\item $A$ is inverse-positive; that is, $A^{-1}$ exists and $A^{-1}\ge0$.
\item $A+\alpha\Id$ is nonsingular for all $\alpha\ge0$.
\end{enumerate}
\end{theorem}

We start with the following auxiliary lemma.

\begin{lemma}\label{triangular Z}
A triangular $Z$-matrix $A\in\bR^{n\times n}$ with positive diagonal is an $M$-matrix.
\end{lemma}

\begin{proof}
Let 
$$
A=\begin{pmatrix}
  \,&a_{11}&a_{12}&\cdots&a_{1n} \,{} \\
     &0&a_{22}&\cdots&a_{2n} \,{} \\
     &\vdots&\ddots&\ddots&\vdots\,{} \\
     &0&\cdots&\cdots&a_{nn}\,{}
 \end{pmatrix}
$$
be an upper triangular $Z$-matrix with positive diagonal. Then all of its leading principle minors are positive:
$$
A_{[k]}=\prod_{i=1}^ka_{ii}>0,\mbox{ for }k\in{1,2,\dots,N}.
$$
By Theorem~\ref{mm} (b), $A$ is an $M$-matrix.
\end{proof}

\medskip

It will be convenient to define the matrices
\begin{align*}
\Phi_{ij}:=e^{-\rho|t_{i-1}-t_{j-1}|}\qquad\text{and}\qquad \Psi_{ij}:=1
\end{align*}
for $i,j=1,\dots, N+1$. Recalling that $G(t)=\lambda e^{-\rho t}+\gamma$, we then have
$\Gamma=\lambda\Phi+\gamma\Psi$ and $\Gamma_\theta=\lambda\Phi+\gamma\Psi+2\theta\Id$. Moreover, for any matrix $A$ we let
$$\wt A_{ij}:=\begin{cases}A_{ij}&\text{if $i>j$,}\\
\frac12A_{ij}&\text{if $i=j$,}\\
0&\text{otherwise.}
\end{cases}
$$
Note that this notation is consistent with \eqref{wt Gamma def}, and we get $\wt \Gamma=\lambda\wt\Phi+\gamma\wt\Phi$. We finally define
$$\wh \Phi:=\wt \Phi+\frac12\Id.$$

\begin{lemma}\label{alphaG}
For $\alpha\ge0$, the inverse of the matrix $\wh \Phi+\alpha \Phi$ is given by
$$
{\small\begin{pmatrix}
  \,&\beta&-a^{\frac1N}\mu\beta^2&-a^{\frac2N}\mu\beta^3&\cdots&-a^{\frac{N-1}{N}}\mu\beta^N&-\frac{a\alpha}{1+\alpha}\beta^N \,{} \\
     &-a^{\frac{1}{N}}\beta&(1+(1-a^{\frac{4}{N}})\alpha)\beta^2&-a^{\frac{1}{N}}\mu\nu\beta^3&\cdots&-a^{\frac{N-2}{N}}\mu\nu\beta^N&-a^{\frac{N-1}{N}}\mu\beta^{N} \,{} \\
     &0&-a^{\frac{1}{N}}\beta&(1+(1-a^{\frac{4}{N}})\alpha)\beta^2&\cdots&-a^{\frac{N-3}{N}}\mu\nu\beta^{N-1}&-a^{\frac{N-2}{N}}\mu\beta^{N-1}\,{} \\
     &0&0&\ddots&\ddots&\vdots&\vdots\,{} \\
     &\vdots&\ddots&\ddots&-a^{\frac{1}{N}}\beta&(1+(1-a^{\frac{4}{N}})\alpha)\beta^2&-a^{\frac{1}{N}}\mu\beta^2\,{} \\
     &0&\cdots&\cdots&0&-a^{\frac{1}{N}}\beta&\beta\,{}
 \end{pmatrix}},
$$
where
$$
\beta=\big(1+(1-a^{\frac2N})\alpha\big)^{-1},\qquad\mu=(1-a^{\frac{2}{N}})\alpha,\qquad\nu=(1-a^{\frac{2}{N}})(1+\alpha).
$$
\end{lemma}

\begin{proof}
Let the   matrix in the statement be denoted by $P$. We rewrite $P$ as
$$
P_{ij}=\begin{cases}\beta,&\text{ if $i=j=1$ or $i=j=N+1$};\\
(1+(1-a^{\frac{4}{N}})\alpha)\beta^2,&\text{ if $i=j\in\{2,\dots,N\}$};\\
-a^{\frac1N}\beta,&\text{ if  $i-j=1$};\\
-a^{\frac{j-i}{N}}\beta^{j-i+2}\mu\nu,&\text{ if $j-i\in\{1,\dots,N-2\}$ and $i\neq1$ and $j\neq N+1$};\\
-a^{\frac {j-i}N}\beta^{j-i+1}\mu,&\text{ if $j-i\in\{1,\dots,N-1\}$ and either $i=1$ or $j=N+1$};\\
-\frac{a\alpha}{1+\alpha}\beta^N,&\text{ if $i=1$ and $j=N+1$},\\
0,&\text{ if $i\ge j+2$.}
\end{cases}
$$
On the other hand, the matrix $\wh \Phi+\alpha \Phi$ can be written as
$$
(\wh \Phi+\alpha \Phi)_{ij}=\begin{cases}1+\alpha,&\text{ if $i=j$};\\
\alpha a^{\frac{j-i}{N}},&\text{ if $i<j$};\\
(1+\alpha)a^{\frac{i-j}{N}},&\text{ if $i>j$}.
\end{cases}
$$
Checking 
$$
\sum_{k=1}^{N+1}P_{ik}(\wh \Phi+\alpha \Phi)_{kj}=\sum_{k=1}^{N+1}(\wh \Phi+\alpha \Phi)_{ik}P_{kj}=\delta_{ij}
$$
for all $i$ and $j$ completes the proof.
\end{proof}

\medskip

Let 
$$\wh \Psi:=\wt \Psi^\top-\frac12\Id.$$

\medskip

\begin{lemma}\label{lm M aux}
The matrix $\Phi^{-1}\big(\wh{\Phi}-\frac{\gamma}{\lambda}\wh{\Psi}\big)$ is a $Z$-matrix and a nonsingular $M$-matrix.
\end{lemma}

\begin{proof} It was shown in~\cite[Theorem 3.4]{AFS1} that
\begin{align}\label{Phi inverse formula}
\Phi^{-1}=\frac{1}{1-a^{\frac{2}{N}}}\begin{pmatrix}
  \,&{1}&{-a^{\frac{1}{N}}}&0&\cdots&\cdots&0 \,{} \\
     &{-a^{\frac{1}{N}}}&{1+a^{\frac{2}{N}}}&{-a^{\frac{1}{N}}}&0&\cdots&0\,{} \\
     &0&\ddots&\ddots&\ddots&\ddots&\vdots\,{} \\
     &\vdots&\ddots&\ddots&\ddots&\ddots&\vdots\,{} \\
     &\vdots&\ddots&\ddots&{-a^{\frac{1}{N}}}&{1+a^{\frac{2}{N}}}&{-a^{\frac{1}{N}}}\,{} \\
     &0&\cdots&\cdots&0&{-a^{\frac{1}{N}}}&{1}\,{}
 \end{pmatrix}.
\end{align}
The matrix $\wh{\Phi}-\frac{\gamma}{\lambda}\wh{\Psi}$ is equal to
$$\begin{pmatrix}
  \,&{1}&{-\frac{\gamma}{\lambda}}&-\frac{\gamma}{\lambda}&\cdots&\cdots&-\frac{\gamma}{\lambda} \,{} \\
     &{a^{\frac{1}{N}}}&{1}&{-\frac{\gamma}{\lambda}}&\cdots&\cdots&-\frac{\gamma}{\lambda}\,{} \\
     &a^{\frac{2}{N}}&\ddots&\ddots&\ddots&\ddots&\vdots\,{} \\
     &\vdots&\ddots&\ddots&\ddots&\ddots&\vdots\,{} \\
     &\vdots&\ddots&\ddots&{a^{\frac{1}{N}}}&{1}&{-\frac{\gamma}{\lambda}}\,{} \\
     &a&\cdots&\cdots&a^{\frac{2}{N}}&{a^{\frac{1}{N}}}&{1}\,{}
 \end{pmatrix}.
$$
A straightforward computation now yields that the matrix $(1-a^{\frac2N})\Phi^{-1}\big(\wh{\Phi}-\frac{\gamma}{\lambda}\wh{\Psi}\big)$ is equal to
$$
\tiny{\begin{pmatrix}
  \,&1-a^{\frac2N}&-a^{\frac1N}-\frac{\gamma}{\lambda}&-(1-a^{\frac{1}{N}})\frac{\gamma}{\lambda}&-(1-a^{\frac{1}{N}})\frac{\gamma}{\lambda}&\cdots&-(1-a^{\frac{1}{N}})\frac{\gamma}{\lambda} \,{} \\
     &0&1+a^{\frac1N}\frac{\gamma}{\lambda}&-a^{\frac1N}-(1-a^{\frac1N}+a^{\frac2N})\frac{\gamma}{\lambda}&-(1-a^{\frac{1}{N}})^2\frac{\gamma}{\lambda}&\cdots&-(1-a^{\frac{1}{N}})^2\frac{\gamma}{\lambda} \,{} \\
     &0&0&1+a^{\frac1N}\frac{\gamma}{\lambda}&-a^{\frac1N}-(1-a^{\frac1N}+a^{\frac2N})\frac{\gamma}{\lambda}&\cdots&-(1-a^{\frac{1}{N}})^2\frac{\gamma}{\lambda}\,{} \\
     &0&\ddots&\ddots&\ddots&\ddots&\vdots\,{} \\
     &\vdots&\ddots&\ddots&\ddots&\ddots&\vdots\,{} \\
     &\vdots&\ddots&\ddots&\ddots&\ddots&-(1-a^{\frac{1}{N}})^2\frac{\gamma}{\lambda}\,{} \\
     &\vdots&\ddots&\ddots&0&1+a^{\frac1N}\frac{\gamma}{\lambda}&-a^{\frac1N}-(1-a^{\frac1N}+a^{\frac2N})\frac{\gamma}{\lambda}\,{} \\
     &0&\cdots&\cdots&0&0&1+a^{\frac1N}\frac{\gamma}{\lambda}\,{}
 \end{pmatrix}},
$$
which is an upper triangular $Z$-matrix with positive diagonal. By Lemma~\ref{triangular Z}, $\Phi^{-1}\big(\wh{\Phi}-\frac{\gamma}{\lambda}\wh{\Psi}\big)$ is hence a nonsingular $M$-matrix.

\end{proof}

\begin{lemma}\label{lm M}
For $\delta\ge0$ the matrix $\Lambda_{\delta} :=\Phi^{-1}\big(\wh{\Phi}-\frac{\gamma}{\lambda}\wh{\Psi}\big)+\delta \Phi^{-1}$ is a nonsingular $M$-matrix.\end{lemma}

\begin{proof}For $\delta=0$ the result follows from Lemma~\ref{lm M aux}. So let us assume henceforth that $\delta>0$. Note first that $\Lambda_{\delta}$ is a $Z$-matrix since both $\Phi^{-1}\big(\wh{\Phi}-\frac{\gamma}{\lambda}\wh{\Psi}\big)$ and $\Phi^{-1}$ are $Z$ matrices by 
 Lemma~\ref{lm M aux}
and \eqref{Phi inverse formula}, respectively. Hence condition (d) of Theorem~\ref{mm} will imply that  $\Lambda_{\delta}$ is a nonsingular $M$-matrix as soon as we can show that  $\Lambda_{\delta}+\alpha\Id$ is invertible for all $\alpha\ge0$.
  
 In a first step, we note  that taking $\gamma=0$ in Lemma~\ref{lm M aux} yields that $\Phi^{-1}\wh \Phi$ is a nonsingular $M$-matrix. Hence $(\alpha\Id+\Phi^{-1}\wh \Phi)^{-1}\ge0$ for all $\alpha\ge0$. 
 It follows that 
$$
\big(\wh \Phi+\alpha \Phi\big)^{-1}\bm 1=\big(\Id+(\alpha \Phi)^{-1}\wh \Phi\big)^{-1}(\alpha \Phi)^{-1}\bm 1=\Big(\alpha\Id+\Phi^{-1}\wh \Phi\Big)^{-1}\Phi^{-1}\bm 1>0,
$$
because by~\cite[Example 3.5]{AFS1},
\begin{align}\label{Phi-11ge0}
{{\Phi}}^{-1}{\bm 1}=\frac{1}{1+a^{\frac1N}}\Big(1,1-a^{\frac1N},\dots,1-a^{\frac1N},1\Big)^T\gg0.
\end{align}
 Since moreover $(\wh \Phi+\alpha \Phi)^{-1}$ is a $Z$-matrix by Lemma~\ref{alphaG}, it follows that $(\wh \Phi+\alpha \Phi)^{-1}$ is a diagonally dominant $Z$-matrix for all $\alpha\ge0$.

In the next step, we show that the matrix 
$$
Q:=(\wh \Phi+\alpha \Phi)^{-1}\Big(\delta\Id-\frac{\gamma}{\lambda}\wh \Psi\Big)
$$
is a  $Z$-matrix. Denoting again $P:=(\wh \Phi+\alpha \Phi)^{-1}$, we get
$$
Q_{ij}=
\delta P_{ij}-\frac{\gamma}{\lambda}\sum_{k=1}^{j-1}P_{ik},
$$
with the convention that $\sum_{k=1}^0a_k=0$.
It follows that $Q_{ii}\ge0$ for all $i$, because 
$P_{ii}\ge0$ and $\frac{\gamma}{\lambda}\sum_{k=1}^{i-1}P_{ik}\le0$ by the fact that $P$ is a $Z$-matrix.
Since $P$ is  diagonally dominant, we have $\sum_{k=1}^{j-1}P_{ik}\ge0$ for any $j>i$ and hence $Q_{ij}=\delta P_{ij}-\frac\gamma\lambda\sum_{k=1}^{j-1}P_{ik}\le0$ for $j>i$. Using the fact that $P_{ik}=0$ for $ k\le i-1$, we get that for $j<i$
$$
Q_{ij}=
\delta P_{ij}-\frac{\gamma}{\lambda}\sum_{k=1}^{j-1}P_{ik}=\delta P_{ij}\le0.
$$
This shows that $Q$ is a $Z$-matrix.

We show next that $Q$ is a nonsingular $M$-matrix. To this end, we note first that  the triangular matrix $\Big(\delta\Id-\frac{\gamma}{\lambda}\wh \Psi\Big)$ is invertible under our assumption $\delta>0$. As a matter of fact, an easy calculation verifies that its inverse is given by
$$
\frac{1}{\delta}\begin{pmatrix}
  \,&1&\sigma&\sigma(1+\sigma)&\cdots&\sigma(1+\sigma)^{N-2}&\sigma(1+\sigma)^{N-1} \,{} \\
     &0&1&\sigma&\sigma(1+\sigma)&\cdots&\sigma(1+\sigma)^{N-2} \,{} \\
     &0&0&1&\sigma&\cdots&\sigma(1+\sigma)^{N-3}\,{} \\
     &0&\ddots&\ddots&\ddots&\ddots&\vdots\,{} \\
     &\vdots&\ddots&\ddots&\ddots&1&\sigma\,{} \\
     &0&\cdots&\cdots&\cdots&0&1\,{}
 \end{pmatrix}\ge0,
$$
where $\sigma:=\frac{\gamma}{\lambda\delta}>0$. Hence,
$$
Q^{-1}=\Big(\delta\Id-\frac{\gamma}{\lambda}\wh \Psi\Big)^{-1}(\wh \Phi+\alpha \Phi)\ge0.
$$
So Theorem~\ref{mm} (c)  {shows} that $Q$ is a nonsingular $M$-matrix.

For the final step, we note first that Theorem~\ref{mm}  (d) implies that 
$
\Id+Q$
is a nonsingular $M$-matrix. In particular, $(\Id+Q)^{-1}$ exists, and so we can define the matrix 
\begin{align*}
 (\Id+Q)^{-1}(\wh \Phi+\alpha \Phi)^{-1}\Phi&= \Big(\delta\Id+\wh \Phi+\alpha \Phi-\frac{\gamma}{\lambda}\wh \Psi\Big)^{-1}\Phi\\
&=\Big(\Id+(\delta \Phi^{-1})^{-1}\Big(\Phi^{-1}\Big(\wh{\Phi}-\frac{\gamma}{\lambda}\wh{\Psi}\Big)+\alpha\Id\Big)\Big)^{-1}(\delta \Phi^{-1})^{-1}\\
&=\Big(\delta \Phi^{-1}+\Phi^{-1}\Big(\wh{\Phi}-\frac{\gamma}{\lambda}\wh{\Psi}\Big)+\alpha\Id\Big)^{-1}\\
&=(\Lambda_\delta +\alpha\Id)^{-1}.
\end{align*}
This proves that $\Lambda_\delta +\alpha\Id$ is invertible and the proof is complete.\end{proof}

\medskip

\begin{lemma}\label{Psi inverse lemma}
Let $A$ be an invertible matrix and suppose that $\alpha\in\bR$ is such that $A+\alpha\Psi$ is invertible. Then the vector $A^{-1}\bm1$ is proportional to $(A+\alpha\Psi)^{-1}\bm1$.
\end{lemma}

\begin{proof} Note that $\Psi\bm x$ is proportional to $\bm1$ for any vector $\bm x$. Hence, 
\begin{align*}
(A+\alpha\Psi)A^{-1}\bm1&=(\Id+\alpha \Psi A^{-1})\bm1=(1+\beta)\bm 1
\end{align*}
for some constant $\beta$. Applying $(A+\alpha\Psi)^{-1}$ to both sides of this equation yields the result.
\end{proof}

\medskip

We are now ready to prove  the remaining implication of Theorem~\ref{nonosc thm}. 

\medskip

\begin{proof}[Proof of {\rm (b)}$\Rightarrow${\rm (c)} in Theorem~\ref{nonosc thm}] We need to show that $\bm v$ has only nonnegative components for $\theta\ge\frac{\lambda+\gamma}{4}$. The vector $\bm v$ is proportional to 
$(\Gamma_\theta+\wt\Gamma)^{-1}\bm1$.
When setting
$$\delta:=\frac{4\theta-(\lambda+\gamma)}{2\lambda}\ge0,
$$
we find that
\begin{align*}
\Gamma_\theta+\wt\Gamma-\gamma\Psi&=\lambda\Phi+\gamma\Psi+2\theta\Id+\lambda\wt \Phi+\gamma\wt \Psi-\gamma\Psi\\
&=\lambda\Phi+2\theta\Id+\lambda\wt \Phi-\gamma\wt \Psi^\top=\lambda\Phi+\lambda\Big(\wh \Phi-\frac{\gamma}{\lambda}\wh \Psi+\delta\Id\Big)=\lambda\Phi\big(\Lambda_\delta+\Id\big),
\end{align*}
and we know   from Lemma~\ref{lm M} that the latter matrix is invertible. 
It therefore follows from Lemma~\ref{Psi inverse lemma} that $\bm v$ is proportional to
\begin{align*}
\big(\Phi(\Lambda_\delta+\Id)\big)^{-1}\bm1&=(\Lambda_\delta+\Id)^{-1}\Phi^{-1}\bm 1.
\end{align*}
As noted in \eqref{Phi-11ge0}, we have $\Phi^{-1}\bm 1\gg0$. Moreover, $\Lambda_\delta$, and hence $\Lambda_\delta+\Id$, are nonsingular $M$-matrices by Lemma~\ref{lm M} and Theorem~\ref{mm} (d). Via Theorem~\ref{mm} (c), these facts imply that  $\big(\Phi(\Lambda_\delta+\Id)\big)^{-1}\bm1\ge0$ and in turn that $\bm v\ge0$.\end{proof}

\section{Conclusion and outlook}\label{Conclusion section}

We have considered a Nash equilibrium for two competing agents in a market impact model with general transient price impact. We have seen that without transaction costs both agents engage in a \lq\lq hot-potato game", which has some similarities to certain events  during the Flash Crash that have been reported in~\cite{SEC,Kirilenko}. 
We have then analyzed the behavior of equilibrium strategies as functions of transaction costs, $\theta$, and trading frequency, $N$. In 
Theorem~\ref{nonosc thm} we have determined the critical value of transaction costs at which the equilibrium strategies $\bm v$ and $\bm w$ become buy-only or sell-only. In Section~\ref{transaction tax section}, numerical simulations  have shown that expected costs can be increasing in the trading frequency for small $\theta$, while they generally decrease for sufficiently large $\theta$. We have also seen that the expected costs of both agents can be lower with additional transaction costs than without. These observations provide some support for the common claim that additional  transaction costs can, at least under certain circumstances such as during a fire sale, have a calming effect on financial markets. 

\bigskip

\noindent{\bf Acknowledgement:} The authors thank Ria Grindel, Elias Strehle, and an anonymous referee for  comments that helped to improve  previous versions of the manuscript.

\bibliography{MarketImpact}{}
\bibliographystyle{abbrv}

\end{document}